\documentclass{article}

\usepackage{amsmath, amsthm, amssymb} 
\numberwithin{equation}{section}

\usepackage{amsthm}
\usepackage{hyperref}
\usepackage[all]{xy}
\usepackage[usenames,dvipsnames]{xcolor}

\makeatletter
\newtheoremstyle{indented}
  {6pt}
  {6pt}
  {\addtolength{\@totalleftmargin}{.75em}
   \addtolength{\linewidth}{-.75em}
   \parshape 1 .75em \linewidth}
  {}
  {\bfseries}
  {.}
  {.5em}
  {}
\makeatother

\theoremstyle{indented}

\newtheorem{theo}{Theorem}[section]

\newtheorem{prop}[theo]{Proposition}

\newtheorem{defn}[theo]{Definition}
\newtheorem{eg}[theo]{Example}

\newtheorem{rmk}[theo]{Remark}

\usepackage{graphicx}
\usepackage{amsfonts}

\newcommand{\e}{{\varepsilon}}

\newcommand{\R}{{\mathbb R}}

\newcommand{\T}{{\mathbb T}}

\newcommand{\N}{{\mathbb N}}

\newcommand{\calB}{{\mathcal{B}}}

\newcommand{\calN}{{\mathcal{N}}}

\newcommand{\calF}{{\mathcal{F}}}
\newcommand{\calE}{{\mathcal{E}}}
\newcommand{\calD}{{\mathcal{D}}}
\newcommand{\calL}{{\mathcal{L}}}
\newcommand{\calU}{{\mathcal{U}}}

\newcommand{\calO}{{\mathcal{O}}}

\newcommand{\calJ}{{\mathcal{J}}}
\newcommand{\calW}{{\mathcal{W}}}
\newcommand{\calZ}{{\mathcal{Z}}}

\newcommand{\id}{{\mathbbm{1}}}

\newcommand{\Sym}{{\textup{Sym}}}

\newcommand{\Dens}{{\textup{Dens}}}

\newcommand{\del}{{\partial}}
\newcommand{\delbar}{{\bar\partial}}

\newcommand{\g}{{\mathfrak{g}}}
\newcommand{\n}{{\mathfrak{n}}}
\newcommand{\el}{{\mathfrak{l}}}

\newcommand{\into}{{\ \hookrightarrow\ }}

\newcommand{\Poly}{{\textup{PolyDiff}}}
\newcommand{\Diff}{{\textup{Diff}}}
\newcommand{\loc}{{\textup{loc}}}

\newcommand{\Der}{{\textup{Der}}}

\usepackage{bbm}

\usepackage[scr=boondoxo,scrscaled=1.05]{mathalfa}
\newcommand{\scg}{{\mathscr{g}}}
\newcommand{\sch}{{\mathscr{h}}}
\newcommand{\scv}{{\mathscr{v}}}
\newcommand{\X}{{\mathscr{X}}}
\newcommand{\Y}{{\mathscr{Y}}}

\newcommand{\scL}{{\mathscr{L}}}
\newcommand{\scl}{{\mathscr{l}}}
\newcommand{\V}{{\mathscr{V}}}

\newcommand{\dR}{{\textup{dR}}}
\newcommand{\triv}{{\textup{triv}}}

\newcommand{\LL}{{\mathbb{L}}}
\newcommand{\TT}{{\mathbb{T}}}

\newcommand{\mr}{\mathrm}
\newcommand{\mm}{\mathfrak{m}}

\begin{document}

\title{Degenerate Classical Field Theories and Boundary Theories}
\author{Dylan Butson and Philsang Yoo}
\date{}

\maketitle

\begin{abstract}
We introduce a framework for degenerate classical field theories in the BV formalism, which allows us to discuss many interesting examples of theories which do not admit a Lagrangian description. Further, we study phase spaces and boundary conditions for classical field theories on manifolds with boundary, and from a fixed classical field theory together with a choice of boundary condition, construct a degenerate classical field theory on the boundary. We apply these ideas to many physically interesting examples including the Kapustin--Witten twists of $\calN=4$ supersymmetric Yang--Mills, Chern--Simons theory, the chiral Wess--Zumino--Witten model, chiral Toda theory, and a new three dimensional degenerate classical field theory we call Whittaker theory.
\end{abstract}

\tableofcontents

\newpage 

\section{Introduction}

The BV-BRST formalism for classical and quantum field theory, especially in the topological case, was studied by Schwartz, Kontsevich and many others \cite{Schwarz} \cite{AKSZ}  using ideas from homological algebra in a geometric way. Recently, following this perspective, Costello \cite{CostelloBook} formulated a precise definition of perturbative classical and quantum field theories in the BV formalism, applicable to many examples of physical interest, and in particular well-defined in the necessary infinite dimensional setting. Further, inspired by the work of Beilinson and Drinfeld \cite{BeilinsonDrinfeld} on 2-dimensional conformal field theories and Lurie \cite{LurieCobordism} on topological field theories, Costello and Gwilliam \cite{CostelloGwilliam} formulated a theory of local observables for general classical and quantum field theories; their work constructs, from the differential geometric input data of a classical or quantum field theory, a factorization algebra of $P_0$ or $BD_0$ algebras on the manifold underlying the field theory. This foundational work has led to the exploration of many interesting classical and quantum field theories and their applications to related areas of pure mathematics \cite{CostelloYangian} \cite{GradyGwilliam} \cite{ChanLeungLi} \cite{CostelloLi2} \cite{GradyLiLi} \cite{LiLi} \cite{CostelloM-theory} \cite{GGW} \cite{LiVAQME}.

In this paper, we develop the framework of \emph{degenerate} classical field theories in the BV formalism, extending the work of Costello and Gwilliam to include a much larger class of examples of classical field theories, which need not have a Lagrangian description. Further, we formulate the notions of phase spaces and boundary conditions for classical field theories on manifolds with boundary, and from a classical field theory together with a choice of boundary condition for it, we construct a possibly degenerate classical field theory on the boundary, which we call the \emph{boundary theory}; this theory governs the behaviour, in the presence of the specified boundary condition, of the observables local to the boundary, called the \emph{boundary observables}. Finally, for any degenerate classical field theory on a manifold $N$, we construct a classical field theory on $N\times \R_{\geq 0}$, called the \emph{universal bulk theory}, and a canonical boundary condition for it, such that the corresponding boundary theory is the given degenerate theory; moreover, we show that this bulk theory is indeed universal among all theories on $N\times \R_{\geq 0}$ which yield the given degenerate theory on $N$ as a boundary theory corresponding to some boundary condition.

We also study in detail many interesting examples of these ideas, recovering several important relationships between classical field theories. In particular, we study many theories that have deep connections to representation theory: the Kapustin--Witten twists of $\calN=4$ supersymmetric Yang--Mills, Chern--Simons theory, the chiral Wess--Zumino--Witten model and the corresponding affine Kac--Moody Poisson vertex algebra, and chiral Toda theory and the corresponding classical affine $\calW$-algebra. We also study a new three dimensional field theory, which we call \emph{Whittaker theory}, occurring as the boundary theory for Kapustin--Witten theory, corresponding to the boundary condition $S$-dual to that which yields Chern--Simons \cite{GaiottoWittenBoundary} \cite{GaiottoWittenKnot}, and discuss briefly its relationship to geometric representation theory \cite{GaitsgoryWhittaker}. We hope that the foundational work occurring here, especially once extended to the quantum level, will be used to yield new results, as well as original or better motivated proofs of existing ones, in related areas of mathematics.

\subsection{Summary of Results}

To begin our overview, we describe schematically the theory of phase spaces, boundary conditions and boundary theories, in the language of global derived symplectic geometry of \cite{PTVV}; although we will not formally work in this setting, we will often explain our motivation from the global perspective. Afterwards, we will outline the formulation of perturbative classical field theory used here, and describe our results more precisely.

Consider an $n$-dimensional field theory with a Lagrangian description: for a closed $n$-manifold $M$, we have a space of fields $\calF=\calF(M)$ on $M$, and an action functional $S \colon \calF \to k$. In the classical BV formalism, the moduli space of the solutions to the corresponding Euler--Lagrange equations on a closed $n$-manifold $M$ is the $(-1)$-shifted symplectic space defined by $\calE\calL(M) := \Gamma_{dS} \times_{T^* \calF } \calF$, which is manifestly symplectic as a Lagrangian intersection inside a $0$-shifted symplectic space. This $(-1)$-shifted symplectic structure is an essential part of the description of a classical field theory in the BV formalism, and moreover, is required for quantization. As an example, Chern--Simons theory gives $\calE\calL_{\mr{CS} }(M^3) = \mr{Loc}_G(M) = \underline{\mr{Map}}( M_{\mr{B}},BG  )$, which inherits the AKSZ symplectic structure. 

The situation when the spacetime manifold $M$ has a nonempty boundary is more subtle. In this case, $\calE\calL(M)$ does not admit a canonical $(-1)$-shifted symplectic structure, and hence does not describe a well-defined field theory on the manifold with boundary. However, the phase space of the theory on the boundary $\calE\calL( \del M)$, which is defined to be the space of germs of solutions near the boundary $\del M$, admits a canonical $0$-shifted symplectic structure. For instance, for Chern--Simons theory,  if $\del M^3 = \Sigma$, then one has $\calE\calL_{\mr{CS} }( \Sigma) = \mr{Loc}_G(\Sigma) = \underline{\mr{Map}}( \Sigma_{\mr{B}},BG  )$.

Further, if one considers a Lagrangian $\calB\to \calE\calL(\del M)$, then $\calE\calL( M,\del M ; \calB ) : = \calE\calL(M) \times_{ \calE\calL(\del M) } \calB$ has a canonical $(-1)$-shifted symplectic structure. Thus, a \emph{classical boundary condition} is defined to be a Lagrangian $\calB \to \calE\calL(\del M)$, as this is the appropriate data such that $\calE\calL( M,\del M ; \calB )$, the space of solutions to the Euler--Lagrange equations on $M$ satisfying the given boundary condition $\calB$ at the boundary $\del M$, is again $(-1)$-shifted symplectic and thus describes a well-defined classical field theory.

Now, we restrict our attention to theories which are topological in the direction normal to the boundary of the manifold; our constructions rely essentially on this simplifying assumption throughout. In this case, restricting to a small collar neighbourhood $U_\e$ of $\del M$ in $M$, we know that all functions on $\calE\calL( U_\e,\del M ; \calB ) $ can be chosen to depend only on the values of the fields restricted to the boundary $\del M$, where the fields must satisfy the boundary condition $\calB$. Thus, we expect that $\calB$ should in some sense carry the structure of a field theory, via this relationship to $\calE\calL( U_\e,\del M ; \calB )$. Moreover, it is generally expected that the functions on a Lagrangian in a $0$-shifted symplectic space should carry a homotopy $P_0$ structure. Our essential observation is that this $P_0$ structure describes the information of being a field theory that $\calB$ inherits from its relationship to $\calE\calL( U_\e,\del M ; \calB )$, and $\calB$ together with this $P_0$ structure is an example of a degenerate classical field theory, which we call the boundary theory.

Finally, for each homotopy $P_0$ space $(X,\Pi)$, there is a $0$-shifted symplectic space $Z_\Pi(X)$, called its higher Poisson centre, in which $X$ is a Lagrangian, and such that $X$ inherits its given $P_0$ structure $\Pi$ as a Lagrangian in this space. Given a degenerate classical field theory, we construct a classical field theory of one dimension higher on a manifold $M$ with boundary, such that its phase space $\calE\calL(\del M)$ is given by the higher Poisson centre corresponding to the $P_0$ structure underlying the degenerate classical field theory, and prove that it is universal among theories producing the given degenerate classical field theory as a boundary theory.

We have now described the basic narrative of this paper schematically in global terms. Of course, it is very difficult in general to understand the global derived stacks describing the spaces of solutions to the Euler--Lagrange equations modulo gauge transformations and their shifted symplectic structures for interesting physical examples. Further, the most common method of analysis of field theory at the quantum level is by perturbation theory, which does not require a global description of the space of solutions to the Euler--Lagrange equations. Thus, we will only attempt to study the formal neighbourhood of a given point in this space of solutions.

Formal pointed spaces have been studied extensively algebraically, using ideas from rational homotopy theory, deformation theory and algebraic geometry, with the conclusion that any such space can be described by its $(-1)$-shifted tangent complex at the geometric point, viewed as a homotopy Lie algebra (see e.g.  \cite{Hinich} \cite{Getzler} \cite{Pridham} \cite{LurieDeformation}. We will work with $L_\infty$ algebras as a concrete model for homotopy Lie algebras, and refer to their corresponding formal spaces as \emph{formal moduli problems}. Thus, we will to describe the space $\calE\calL(M)$ of solutions to the Euler--Lagrange equations of a classical field theory as a formal moduli problem $X$ corresponding to an $L_\infty$ algebra $\g=\T_o[-1]X$.

A crucial concept in classical and quantum field theory which can still be understood at the perturbative level is that of local observables, which, at least classically, are simply the functions on the space $\calE\calL(M)$ that depend only on the values of the fields in a given open set $U\subset M$. Since the Euler--Lagrange equations corresponding to a local action functional are differential equations, their solutions restrict, and thus the spaces of solutions form a presheaf $U\mapsto \calE\calL(U)$ on $M$. Correspondingly, the spaces of functions on the spaces of solutions to the Euler--Lagrange equations form a precosheaf $U\mapsto \calO(\calE\calL(U))$, and these spaces of functions $\calO(\calE\calL(U))$ are precisely the spaces of classical observables local to the open sets $U\subset M$.

Combining these insights, we will describe the space of solutions to the Euler--Lagrange equations underlying a classical field theory as a presheaf $\X$ on $M$ of formal moduli problems, or equivalently as the corresponding presheaf of $L_\infty$ algebras $\scg=\T_o[-1]\X$. In fact, the presheaves of $(-1)$-shifted tangent complexes to the solutions to the Euler--Lagrange equations of a local action functional have a natural strict model defined in terms of differential geometry, called a \emph{local $L_\infty$ algebra}, coming from their description as solutions to differential equations on spaces of sections of a vector bundle over $M$; we define \emph{local moduli problems} as the presheaves $\X$ of formal moduli problems corresponding to local $L_\infty$ algebras $\scg$.

In their work on classical field theory, Costello and Gwilliam \cite{CostelloGwilliam} study extensively the geometry of local moduli problems: they formulate an appropriate notion of local $n$-shifted symplectic structure $\omega$ on a local moduli problem $\X$, define a classical field theory as a local moduli problem $\X$ together with a local $(-1)$-shifted symplectic form $\omega$ on $\X$, and prove the expected equivalence between classical field theories in this sense and local action functionals, the usual defining data of a classical field theory. Further, they show that the precosheaf of strict dg $P_0$ algebras $\calO(\X)$ representing the observables of the classical field theory determines a $P_0$ factorization algebra; this is the essential algebraic object they extract from the differential geometric input data of the classical field theory. We begin here by first reviewing these results in Subsections \ref{lla}--\ref{laf}.

Our original work begins in Subsection $\ref{dcft}$ with the formulation of the notion of a local (homotopy) $(-1)$-shifted Poisson structure $\bar\Pi$ on a local moduli problem, and the definition of a degenerate classical field theory as a local moduli problem $\X$ together with a local $(-1)$-shifted Poisson structure $\bar\Pi$. This is a large generalization of the class of classical field theories introduced by Costello and Gwilliam, allowing us to describe many more examples of theories, which need not admit a Lagrangian description. In particular, the boundary theories constructed in the later part of the paper will in general be degenerate classical field theories in this sense. The primary motivating result about our formulation of degenerate classical field theories is the following:

\begin{theo}
Let $(\X,\Pi)$ be a degenerate classical field theory on $M$. Then $\calO(\X)$ determines a $P_0$ factorization algebra on $M$.
\end{theo}

The locality conditions in our definition of shifted Poisson structure, as well as the original definition of local moduli problem, were chosen to ensure this result. We emphasize the factorization structure in this result is a significant amount of additional information; for example, in the case of chiral conformal field theory, this information is equivalent to that of a Poisson vertex algebra or Coisson algebra, rather than just a Poisson algebra.

In Section $\ref{bt}$ of the paper we study classical field theories on manifolds with boundary, their phase spaces, boundary conditions, and induced boundary theories, in this local, formal framework. Given a (non-degenerate) classical field theory $(\X,\omega)$ on a manifold $M$ with boundary $\del M$, we define the phase space of $(\X,\omega)$ on $\del M$ as a local moduli problem $\X^\del$ on $\del M$ together with a local $0$-shifted symplectic form $\omega^\del$ on $\X^\del$ which can be constructed from the data of $(\X,\omega)$; to state this result most easily, we at this point restrict our attention to theories which have the property of being topological in the direction normal to the boundary of the manifold. Setting aside the local differential geometric subtleties of our definitions, the notion of phase space defined here agrees with the definition of boundary BFV theory in \cite{CattaneoMnevReshetikhin}.

Next, we define the notion of a regular embedded, local boundary condition for $(\X,\omega)$: this is a local moduli problem $\scL_+$ on $\del M$, together with a homotopically strict map $\scL_+\into \X^\del$ of local moduli problems such that $\scL_+$ is in a certain sense a Lagrangian in $\X^\del$ with respect to its $0$-shifted symplectic structure. Our main result is then that under these hypotheses, there is a natural local homotopy $(-1)$-shifted Poisson structure:

\begin{theo}
Let $(\X,\omega)$ be a classical field theory on a manifold with boundary $M$, and $\scL_+\into \X^\del$ a regular embedded, local boundary condition for $(\X,\omega)$. Then there is a canonical local $(-1)$-shifted homotopy Poisson structure $\bar\Pi$ on $\scL_+$.
\end{theo}

We define the boundary theory associated to $(\X,\omega)$ and $\scL_+$ as the resulting degenerate classical field theory on $\del M$, in the case where $\bar\Pi$ is in fact a strict Poisson structure. In general, the boundary theory will just be a formal moduli problem equipped with a local $(-1)$-shifted homotopy Poisson structure.

In the final subsection of Section $\ref{bt}$, we formulate the notion of local higher Poisson centre of a local moduli problem over a closed manifold $N$, equipped with a local $(-1)$-shifted homotopy Poisson structure, as another local moduli problem $Z_{\bar\Pi}(\scL)$ on $N$ together with a canonical local, $0$-shifted symplectic structure. Further we construct a (non-degenerate) classical field theory $\calU_{\bar\Pi}(\scL)$ on $N\times \R_{\geq 0}$ which has phase space on $N$ given by the local higher Poisson centre $\calZ_{\bar\Pi}(\scL)$, and a canonical boundary condition given by the natural map $\scL\into \calZ_{\bar\Pi}(\scL)$; by construction, this boundary condition yields the original degenerate classical field theory $\scL$ as its boundary theory, and moreover this theory is the universal such theory yielding $\scL$ as a boundary theory.

In Section $\ref{section:examples}$, we give a detailed discussion of several interesting examples to which we apply this formalism. We begin by discussing topological classical mechanics valued in a $0$-shifted symplectic formal moduli problem $X$; its phase space is just the target formal moduli problem $X$ itself, and its boundary conditions are simply strict derived Lagrangians $L_+\to X$ with underlying vector space map injective. More generally, using the framework discussed above for degenerate classical field theories, we can consider topological classical mechanics valued in a $0$-shifted Poisson formal moduli problem; the corresponding universal bulk theory yields the well-studied Poisson sigma model for the given target.

Next, we discuss 3-dimensional classical (complexified) Chern--Simons theory with gauge group $G$: the phase space of this theory on a manifold $M$, with boundary given by a compact Riemann surface $\Sigma$, is the local, formal analogue of $\mr{Flat}_G(\Sigma)$, which inherits its $0$-shifted symplectic structure as a twisted cotangent bundle of $\mr{Bun}_G(\Sigma)$. One natural choice of Lagrangian in this space is the cotangent fibre over the trivial bundle; the resulting boundary condition yields as a boundary theory the perturbative chiral WZW model, a degenerate classical field theory which has corresponding $P_0$ factorization algebra equivalent to the affine Kac--Moody Poisson vertex algebra.

Another natural Lagrangian in $\mr{Flat}_G(\Sigma)$ is the space $\mr{Op}_G(\Sigma)$ of $G$-opers on $\Sigma$, introduced by \cite{BDopers}. The perturbative incarnation of this Lagrangian gives a boundary condition for Chern--Simons with boundary theory given by chiral Toda theory, another interesting example of a degenerate field theory. The $P_0$ factorization algebra corresponding to classical chiral Toda theory recovers the classical BRST complex for Drinfeld--Sokolov reduction of the affine Kac--Moody Poisson vertex algebra \cite{DrinfeldSokolovI} \cite{DrinfeldSokolovII}, yielding the classical affine $\calW$-algebra, as a Poisson vertex algebra.

Next, we introduce the Kapustin--Witten twist of $\calN=4$ supersymmetric Yang-Mills theory, studied by \cite{KapustinWitten} to explain the geometric Langlands program in terms of quantum field theory, which is in fact the universal bulk theory of Chern--Simons. We describe the Kapustin--Witten theory itself in detail, and explain this claim in our framework. Further, we define in our context a twisted version of a Nahm pole boundary condition for Kapustin--Witten theory studied in \cite{GaiottoWittenBoundary} \cite{GaiottoWittenKnot}, which is S-dual to the boundary condition recovering Chern--Simons theory. We study the boundary theory corresponding to this boundary condition as a three dimensional degenerate classical field theory, which we call the Whittaker theory, and briefly describe its expected role in relation to geometric representation theory.

We wish to again acknowledge the related work of Cattaneo, Mnev, and Reshetikhin \cite{CattaneoMnevReshetikhin} \cite{CattaneoMnevReshetikhinII} on field theories on manifolds with boundary in the BV formalism. While this paper has some overlap with their work, and certainly shares a common general perspective, the technical results and examples presented here are complementary to those occurring in their papers: in \cite{CattaneoMnevReshetikhin} the study of boundary BFV theories, which we call phase spaces, is emphasized. In this paper, we introduce the framework of degenerate field theories, and using this describe the notion of boundary theory presented here.

\subsection{Acknowledgements}

We would like first to thank Kevin Costello, for all his help in developing our perspectives on classical and quantum field theory in general. Further, we would like to thank Ryan Grady and Brian Williams for several useful discussions on this material, as well as their efforts in ongoing collaboration with the authors on extensions of the present work. We also thank Pavel Safronov for useful discussions. We are grateful to Albert Schwarz for his interest, and to Pavel Mnev for helping clarify the relationship of this work with the existing literature.

\newpage

\subsection{Conventions}

Here is a collection of conventions we use throughout the work.

\begin{itemize}
\item A generic spacetime manifold is denoted by $M$. $C^\infty_M$ is the sheaf of smooth functions on $M$, $\calD_M$ is the sheaf of differential operators on $M$, and $\Dens_M$ is the sheaf of densities.
\item A generic space is denoted by $X$. In particular, by abuse of notation, we write a formal moduli problem as $X$. On the other hand, $ \g$ stands for a generic $L_\infty$ algebra. If the two notations appear in the same place, one should think of $\g= \T_o[-1]X$ and $X=B\g$ by the fundamental theorem of deformation theory.
\item For a formal moduli space $X$ and the corresponding $L_\infty$ algebra $\g$, one defines $\T_X $ to be the vector bundle corresponding to the $\g$-module $\g[1]$ and $\LL_X$ to be corresponding to the $\g$-module $\g^*[-1]$.
\item We use the script font to indicate sheafiness on $M$: $\X$ will denote a presheaf of formal moduli problems and $\scg$ a presheaf of $L_\infty$-algebras. For each open set $U\subset M$, we write the assignments as $\X_U$ and $\scg_U$, respectively.

\item For $\calE,\calE_i,\calF$ sheaves of sections of vector bundles, we let $\Diff(\calE,\calF)$ and $\Poly(\calE_1\otimes ...\otimes \calE_n,\calF)$ denote the sheaves of differential and polydifferential operators. The latter is defined by
$$\Poly(\calE_1\otimes ...\otimes \calE_n,\calF)=\Diff(\calE_1,C^\infty_M)\otimes_{C^\infty_M} ... \otimes_{C^\infty_M} \Diff(\calE_n,C^\infty_M)\otimes_{C^\infty_M} \calF$$

\item For $\scg$ the sheaf of sections of a vector bundle $L$, for example those underlying the presheaves of $L_\infty$ algebras above, we will introduce several notations:
\subitem $\bar\scg$ will denote the sheaf of distributional sections of $L$.
\subitem $\scg_c$ will denote the cosheaf of compactly supported sections of $L$.
\subitem $\scg^\vee$ will denote the sheaf of sections of the dual bundle $L^\vee$.
\subitem $\scg^!:=\scg^\vee\otimes_{C^\infty_M} \Dens_M$ will denote the sheaf of sections of the Verdier dual bundle $L^!=L^\vee\otimes \Dens_M$.
\subitem $\scg^*=\bar\scg^!_c$ will denote the cosheaf of compactly supported, distributional sections of the Verdier dual bundle $L^!$. Note that on each open set this gives the linear dual of $\scg$ in the category of nuclear Fr\'echet spaces.
\subitem $\calJ(\scg)$ will denote the sheaf of sections of the infinite jet bundle $J(L)$.

\item The notation $\otimes$ without subscript when applied to infinite dimensional vector spaces denotes the completed, projective tensor product in the category of nuclear Fr\'echet spaces.
\end{itemize}

\newpage

\section{Classical Field Theories \label{ClFT}}

In this section, we begin by reviewing the work of Costello \cite{CostelloBook} and Costello and Gwilliam \cite{CostelloGwilliam} on classical field theories: We first define local $L_\infty$ algebras and their corresponding local moduli problems, the primary objects of study in their formulation of classical field theory. Then, we define local $L_\infty$ modules and vector bundles on local moduli problems, their geometric counterpart. Further, we recall the notion of local sections of a vector bundle on a local moduli problem, and explain that the defining data of a local $L_\infty$ algebra can be understood in terms of local vector fields on a formal moduli problem. Next, we recall the definition of a strictly local $n$-shifted symplectic structure, and of a classical field theory in the BV formalism in these terms, as well as the main result of \cite{CostelloGwilliam} on classical field theories, which constructs a $P_0$ factorization algebra of local observables from the differential geometric input data of a classical field theory. Finally, we state the expected equivalence between classical field theories in this sense and local action functionals satisfying the classical master equation.

In the final subsection of this section, we begin our original work by defining local (homotopy) $(-1)$-shifted Poisson structures on local moduli problems and then defining degenerate classical field theories in these terms. Our first main result is Theorem $\ref{dftf}$, which states that the local observables of degenerate classical field theory also form a $P_0$ factorization algebra.

\subsection{Local $L_\infty$ Algebras and Local Moduli Problems \label{lla}}

As we have discussed in the introduction, we will describe the space of solutions to the Euler--Lagrange equations of a classical field theory in terms of a presheaf of formal moduli problems. Further, it will be convenient to describe each of the formal moduli problems $\X_U$ in terms of its $(-1)$-shifted tangent complex $\scg_U=\T_o[-1]\X_U$. The $L_\infty$ algebra structure on $\scg_U$ encodes all the geometric information about the formal moduli problem $\X_U$, and our calculations will often be most easily understood in these terms. Moreover, the locality constraint on the presheaf $\X$ of formal moduli problems coming from the Euler--Lagrange equations of a classical field theory is easily phrased in terms of the corresponding presheaf of $(-1)$-shifted tangent complexes $\scg$ as follows:

\begin{defn} A \emph{local $L_\infty$ algebra} on $M$ is a smooth graded vector bundle $L$ on $M$, with sheaf of sections $\scg$, together with a collection of polydifferential operators 
$$\{l_n:\scg^{\otimes n} \to \scg[2-n]\}_{n\in\N}$$
making $\scg$ into a presheaf on $M$ of $L_\infty$ algebras.

A \emph{local moduli problem} $\X$ on $M$ is a presheaf of formal moduli problems with presheaf of $(-1)$-shifted tangent complexes $\T_o[-1]\X$ modeled by a local $L_\infty$ algebra $\scg$ on $M$.

A local $L_\infty$ algebra $\scg$ is called \emph{abelian} if $l_n=0$ for $n\geq 2$ and \emph{trivial} if $l_n=0$ for $n\geq 1$.
\end{defn}

For each local moduli problem $\X$ with corresponding local $L_\infty$ algebra $\T_o[-1]\X=\scg$, we let $\calO(\X)\equiv C^\bullet(\scg)$, the Chevalley--Eilenberg cochains on $\scg$, be the precosheaf of cdgas defined by:
$$\calO(\X)_U=C^\bullet(\scg)_U= \widehat\Sym^{\bullet}(\scg_U^*[-1])$$
equipped with the Chevalley--Eilenberg differential, defined by precomposition with the polydifferential operators defining the $L_\infty$ structure maps of $\scg$.

Many explicit examples of local moduli problems of significant physical and mathematical interest will be discussed in Section $\ref{section:examples}$. Throughout the rest of the text we will consistently discuss only a few very simple examples, and some generalities on classes of examples.

\begin{eg} Let $M=\{*\}$ be a point. Then all of the differential geometric subtleties of the definitions above hold vacuously, and thus a local $L_\infty$ algebra over $M$ is just a finite type $L_\infty$ algebra and a local moduli problem over $M$ is just a finite type formal moduli problem.
\end{eg}

\begin{eg} Let $E$ be a vector bundle over $M$ with sheaf of sections $\calE$ and $L=E[-1]$ with sheaf of sections $\scg$. Then we can consider $\scg$ a trivial $L_\infty$ algebra with corresponding local moduli problem $\X$ satisfying
$$\calO(\X)_U= \widehat\Sym^\bullet( \calE_U^*)$$
which is simply the space of formal power series on $\calE_U$. Thus, we can identify $\X_U$ with the affine space $\calE_U$, and thus think of $\X$ as the sheaf of infinite dimensional affine spaces given by $\calE$ itself.

Similarly, given an abelian $L_\infty$ algebra $\scg=\calE[-1]$, we can identify the corresponding local moduli problem $\X$ with the sheaf of dg affine spaces $\calE$.
\end{eg}

\begin{eg} Let $M$ be a smooth manifold, $X$ be a complex manifold, and $\g$ be a finite type $L_\infty$ algebra with $L_\infty$ structure maps $\{l_n\}_{n\in\N}$. We define local $L_\infty$ algebras over $M$ and $X$, respectively, by
\begin{align*} \T_o[-1]\underline{\mr{Map}}(M_\dR,B\g):= \Omega^\bullet_M\otimes \g  && l_1=d_M\otimes \id_\g + \id_{\Omega^\bullet_M}\otimes l_1^\g\ , \quad l_n= \mu^n\otimes l_n^\g \quad\text{for $n\geq 2$}  \\
\T_o[-1]\underline{\mr{Map}}(X_\delbar,B\g):= \Omega^{0,\bullet}_X\otimes \g && l_1=\delbar_X \otimes \id_\g + \id_{\Omega^{0,\bullet}_X}\otimes l_1^\g\ , \quad l_n= \mu^n\otimes l_n^\g \quad\text{for $n\geq 2$} 
\end{align*}
where $\mu^n:(\Omega^\bullet_M)^{\otimes n}\to \Omega^\bullet_M$ denotes the $n$-ary algebra multiplication, and similarly for $\Omega_X^{0,\bullet}$. We denote the corresponding local moduli problems by
$$\widehat{\underline{\mr{Map}}}_o(M_\dR,B\g) \quad\quad\text{and}\quad\quad \widehat {\underline{\mr{Map}}}_o(X_\delbar,B\g) $$ 
respectively. We call such local moduli problems \emph{formal mapping spaces}; these include the local moduli problems underlying AKSZ type classical field theories $\cite{AKSZ}$.

More generally, let $\Y$ be a local moduli problem over a manifold $N$ and $\scg=\T_o[-1]\Y$ the corresponding local $L_\infty$ algebra. We define the local $L_\infty$ algebras over $M\times N$ and $X\times N$, respectively, by
\begin{align*} \T_o[-1]\underline{\mr{Map}}(M_\dR,\Y):= \Omega^\bullet_M\otimes \scg  && l_1=d_M\otimes \id_\scg + \id_{\Omega^\bullet_M}\otimes l_1^\scg\ , \quad l_n= \mu^n\otimes l_n^\scg \quad\text{for $n\geq 2$}  \\
\T_o[-1]\underline{\mr{Map}}(X_\delbar,\Y):= \Omega^{0,\bullet}_X\otimes \scg && l_1=\delbar_X \otimes \id_\scg + \id_{\Omega^{0,\bullet}_X}\otimes l_1^\scg\ , \quad l_n= \mu^n\otimes l_n^\scg \quad\text{for $n\geq 2$} 
\end{align*}
We denote the corresponding local moduli problems on $M\times N$ and $X\times N$ by
$$\widehat {\underline{\mr{Map}}}_o(M_\dR,\Y) \quad\quad\text{and}\quad\quad \widehat {\underline{\mr{Map}}}_o(X_\delbar,\Y) ,$$
respectively. Note that we are, by abuse of notation, using $\Omega^\bullet_M\otimes \scg$ to denote the sheaf of sections of the vector bundle $\wedge^\bullet T^*_M\boxtimes L$ over $M\times N$, and that the local $L_\infty$ structure maps as defined are indeed polydifferential operators on this bundle, as required; this holds for $\Omega^{0,\bullet}_X\otimes \scg$ as well.

In Subsection $\ref{slss}$ we will discuss the AKSZ construction of shifted symplectic structures on such formal moduli problems.
\end{eg}

We now give a remark on the comparison between the notion of formal mapping spaces in terms of local moduli problems given above in terms of smooth differential geometry, and the global derived algebraic geometry approach to mapping stacks and various spaces defined in terms of them. In particular, this comparison will be important for understanding the global motivation for many of the examples discussed in Section $\ref{section:examples}$.

\begin{rmk}\label{rmk:algebraic vs smooth}
Although we work primarily in categories defined in terms of smooth differential geometry, there is also a global algebraic framework for discussing field theories; see for example $\cite{PTVV}$. This framework can be used in special situations to capture the global algebraic structure of the moduli spaces of the solutions to the equations of motion, when such a subtle algebraic structure needs to be considered, for instance, in the context of the geometric Langlands program \cite{EY}. Of course, a generic field theory does not have any inherent global algebraic structure, but such information, when it does exist, is often crucial for understanding relationships between field theory and pure mathematics.

The underlying spaces of interest for a field theory are usually a real manifold $M$ or a complex manifold $X$, which we assume is the complex points of a smooth variety $X$, to facilitate comparison with the algebraic framework. In derived algebraic geometry, one defines the Betti stack $M_\mr{B}$, the Dolbeault stack $X_{\mr{Dol}}$ and the de Rham stack $X_{\mr{DR}}$. On the other hand, in the smooth category, one works with the ringed spaces $M_{\mr{sm}}  = (M ,C^\infty_M) $, $M_{\mr{dR}} = (M, (\Omega^\bullet_M , d_{M}))$, and $X_{\delbar } = (X, (\Omega^{0,\bullet}_X , \delbar_X ))$; see for example \cite{CostelloWittenGenusII}, \cite{GradyGwilliamLinfinity}. We list here the analogous pairs of spaces in the two frameworks, along with several moduli spaces of $G$-bundles one can consider on them: \\

\begin{tabular}{|c|c|}
\hline
 global algebraic & smooth formal \\
\hline
$M_{\mr{B}}$ & $M_{\mr{dR}}$\\[0.5ex]
\hline
 $X$ & $X_{\delbar }$  \\[0.5ex]
\hline
  $X _{\mr{Dol}} $ & $(X,  (\Omega^{\bullet,\bullet}_X , \delbar_X ) )$ \\[0.5ex]
  \hline
  $X _{\mr{DR}} $ & $ (X,  (\Omega^{\bullet,\bullet}_X , d_X ) )$\\[0.5ex]
\hline
\end{tabular} \qquad
\begin{tabular}{|c|c|}
\hline
 global algebraic & smooth formal \\
\hline
 $ \mr{Loc}_G(M)= \underline{\mr{Map}}(M_{\mr{B}}, \mr{BG} )$ & $\scg=(\Omega^{\bullet}_M\otimes \g ,  d_M )$  \\ [0.5ex]
\hline
  $ \mr{Bun}_G(X) =  \underline{\mr{Map}}(X, \mr{BG} )$ & $\scg=(\Omega^{0,\bullet}_X\otimes \g , \delbar_X )$  \\[0.5ex]
\hline
  $\mr{Higgs}_G(X) =  \underline{\mr{Map}}(X_{\mr{Dol}}, \mr{BG} )$ & $\scg=(\Omega^{\bullet,\bullet}_X\otimes \g , \delbar_X) $\\[0.5ex]
  \hline
$\mr{Flat}_G(X) =  \underline{\mr{Map}}(X_{\mr{DR}}, \mr{BG} )$ & $\scg=(\Omega^{\bullet,\bullet}_X\otimes \g ,  d_X )$ \\[0.5ex]
\hline
\end{tabular}\\

Note that in the smooth formal description, $\mr{Loc}_G(X)$ and $\mr{Flat}_G(X)$ have the same presentation, because $(X,  (\Omega^{\bullet,\bullet}_X , d_X ) )$ can be written as $X_{\mr{dR}}$ if $X$ is regarded as a real manifold: this is the reason why we might want to use the algebraic language when describing the global moduli space and why one can freely decide the global algebraic model while working with a fixed smooth formal model. On the other hand, $M_{\mr{dR}}$ in the smooth category clearly has more information than the one of $M_{\mr{B}}$ from the homotopy type of $M$, which is crucial for our discussion throughout. Also, $X_{\mr{dR}}$ has much better flexibility than $X_{\mr{DR}}$ which is purely of an algebraic nature.

In our text, we use advantages of both presentations in the following way. First of all, from the global algebraic description of the moduli space, using the adjunction of the mapping stack $\underline{\mr{Map}}(X \times Y,Z) = \underline{ \mr{Map} }( X, \underline{\mr{Map}}(Y,Z)  )$, one can always move the algebraic dependence to the target. Often times the remaining dependence on the domain is given by either $X_{\mr{B}}$ or $X_{\mr{DR}}$, in which case we read it as $X_{\mr{dR}}$. In this way, we can simultaneously capture the algebraic dependence of interest and retain the flexibility we need to work in the smooth formal description of field theory. Many explicit examples will be discussed in Section \ref{section:examples}.
\end{rmk}

Finally, we define the notion of maps of local $L_\infty$ algebras and correspondingly of local moduli problems:

\begin{defn} Let $\scg,\sch$ be local $L_\infty$ algebras on $M$ with underlying vector bundles $L,H$. A \emph{homotopically strict, strictly local map} $f:\scg\to\sch$ of local $L_\infty$ algebras is a map $L\to H$ of vector bundles on $M$ such that the induced map on sections $f:\scg\to \sch$ satisfies
$$l^\sch_n\circ f^{\otimes n}= f\circ l^\scg_n$$
for each $n\in\N$, where $\{l^\scg_n\}_{n\in\N},\{l^\sch_n\}_{n\in\N}$ are the $L_\infty$ brackets of $\scg,\sch$, respectively.

Let $\X,\mathscr{Y}$ be local moduli problems. A \emph{homotopically strict, strictly local map of local moduli problems} $F:\X\to \Y$ is a homotopically strict, strictly local map of local $L_\infty$ algebras $f:\T_o[-1]\X\to \T_o[-1]\Y$.
\end{defn}
Note that the term homotopically strict refers to the fact that the maps of $L_\infty$ algebras are strict maps in the homotopical sense, and the term strictly local refers to the fact that the maps are required to be built from bundle maps rather than from arbitrary differential operators. In subsequent work, as well as perhaps in later updates of the present paper, more general notions of maps of local $L_\infty$ algebras, and their corresponding maps of formal moduli problems, will appear.

\subsection{Local $L_\infty$ Modules and Vector Bundles on Local Moduli Problems}

In this subsection, we introduce the notions of local $L_\infty$ modules and vector bundles on local moduli problems, discuss various spaces of sections of these vector bundles, define the tangent and cotangent bundles to a local moduli problem and discuss their natural geometric features.

A vector bundle over a local moduli problem should give, for each open set $U$ of $M$, a derived space of sections $\Gamma(\X_U,\V_U)$ which is a dg-module of $\calO(\X)_U$ in a local way. The correct model for this structure is the following:

\begin{defn} Let $\scg$ be a local $L_\infty$ algebra. A \emph{local $L_\infty$ module} for $\scg$ is a smooth graded vector bundle $V$ on $M$ with sheaf of sections $\scv$, together with a differential operator $d:\scv\to \scv[1]$ satisfying $d^2=0$, and a collection of polydifferential operators
$$\{m_n:\scg^{\otimes n}\otimes \scv\to \scv[1-n]\}_{n\geq 0}$$
making $\scv$ into a presheaf on $M$ of $L_\infty$ algebra modules for the presheaf of $L_\infty$ algebras $\scg$.

Let $\X$ be a local moduli problem and $\scg=\T_o[-1]\X$. We define a \emph{vector bundle} $\V$ on $\X$ as a local $L_\infty$ module $\scv$ for $\scg$.
\end{defn}

Given a local $L_\infty$ algebra $\scg$ and a local $L_\infty$ module $\scv$ for $\scg$, we define $C^\bullet(\scg;\scv)$ and $\bar C_c^\bullet(\scg;\scv)$, the spaces of mollified and general Chevalley--Eilenberg cochains on $\scg$ with coefficients in $\scv$, to be the assignments to each $U\subset M$ the $\calO(\X)_U$ dg-modules defined by
$$C^\bullet(\scg; \scv)_U=\widehat\Sym^\bullet(\scg_U^*[-1])\otimes \scv_U \quad\quad \text{and}\quad\quad \bar C_c^\bullet(\scg;\scv)_U=\widehat\Sym^\bullet(\scg_U^*[-1])\otimes  (\bar \scv_c)_U$$
respectively, each equipped with their Chevalley--Eilenberg differential, defined by precomposition with the polydifferential operators defining the structure maps of the local $L_\infty$ algebra and postcomposition from the structure maps of the local $L_\infty$ module.

\begin{defn} Let $\X$ be a local moduli problem, $\scg=\T_o[-1]\X$, and $\V$ a local vector bundle on $\X$ corresponding to a local $L_\infty$ module $\scv$ for $\scg$.

The families $\Gamma(\X,\V)$ and $\bar\Gamma_c(\X,\V)$ of spaces of \emph{mollified} and \emph{general sections} of the local vector bundle $\V$ on $\X$ assign to each $U\subset M$ the $\calO(\X)_U$ dg-modules
$$\Gamma(\X,\V)_U=C^\bullet(\scg;\scv)_U\quad\quad\text{and}\quad\quad \bar\Gamma_c(\X,\V)_U= \bar C_c^\bullet(\scg;\scv)_U$$
each equipped with their Chevalley--Eilenberg differentials.
\end{defn}

Note that $\Gamma(\X,\V)$ does not define a presheaf or precosheaf on $M$ in general, and indeed this is not our expectation. We will see below that the total space of $\V$ will itself define a local moduli problem, and the corresponding presheaf of spaces on $M$ can be understood as a presheaf valued in the category of vector bundles over formal moduli problems, covering the presheaf $\X$ of formal moduli problems; however, the structure maps for this presheaf are bundle maps covering non-trivial maps of spaces, which needn't induce maps on spaces of sections.

As usual, applying functorial vector space operations to a vector bundle $\V$ on a local moduli problem $\X$ yields new vector bundles: the categories of smooth graded vector bundles on $M$ and usual $L_\infty$ modules both admit direct sums, degree shifts, and tensor, symmetric and alternating products, and each in a compatible way, so that these operations extend naturally to vector bundles over local moduli problems; these operations are denoted as usual.

\begin{eg}
Let $\X$ a local moduli problem with $\scg=\T_o[-1]\X$, and let $L$ denote the smooth graded vector bundle underlying the local $L_\infty$ algebra $\scg$. Then the $L_\infty$ structure maps for $\scg$ define a $\scg$-module structure on $\scg[1]$. We define the corresponding vector bundle on $\X$ to be the tangent bundle and denote it by $\T_\X$. Thus, we have
$$\Gamma(\X,\T_\X) = C^\bullet(\scg;\scg[1]). $$
Further, for each $U\subset M$ the vector fields on $\X_U$ yield infinitesimal automorphisms of $\X_U$: we have a map
$$\Gamma(\X,\T_\X[k])_U\to \Der^k(\calO(\X)_U) \quad\quad \text{defined by}\quad\quad  X(f)=\langle X,d_\dR f\rangle_{\T_\X}$$ 
for $f\in \calO(\X)_U$ and $X\in\Gamma(\X,\T_\X[k])_U$, where $\Der^k$ denotes the cohomological degree $k$ derivations of the cdga $\calO(\X)_U$, and where $d_\dR$ and $\langle\ ,\ \rangle_{\T_\X}$ are defined below, in this subsection.
\end{eg}

Since the relevant categories of vector spaces are infinite dimensional, one must be a bit more careful about defining the dual bundle of a bundle $\V$ over a local moduli problem $\X$. We define the local Verdier dual vector bundle $\V^!$ on $\X$ as corresponding to the $L_\infty$ module $\scv^!$; this is defined as having underlying vector bundle $V^!=V^\vee\otimes \Dens_M$ together with
$$\{m_n^!:\scg^{\otimes n}\otimes \scv^!\to \scv^![1-n]\}_{n\geq 0}$$
which are defined by taking the formal adjoint, in the $\scv$ variables, of the polydifferential operators defining the local $L_\infty$ module structure maps for $\V$. One sense in which this is an appropriate notion of dual bundle is that there exists a pairing
$$\langle\cdot,\cdot\rangle_\V:\Gamma(\X,\V)\otimes\bar\Gamma_c(\X,\V^!)\to \calO(\X) $$
which is non-degenerate on each $U\subset M$, defined by the duality pairing $\scv\otimes \bar\scv_c^!\to k$.

\begin{eg} We define the cotangent bundle to a local moduli problem $\X$ by $\LL_\X=\T_\X^!$. In particular, we have $\Gamma(\X,\LL_\X)=C^\bullet(\scg;\scg^![-1])$. Further, there exists a map of families of $\calO(\X)$-dg-modules
$$d_\textup{dR}:\calO(\X)\to \bar\Gamma_c(\X, \LL_\X)$$
giving the family of de Rham differentials on the local moduli problem $\X$ over $M$, defined degree-wise by the inclusion $\Sym^n(\scg^*[-1])\into \Sym^{n-1}(\scg^*[-1])\otimes \scg^*[-1]$.
\end{eg}

Next, we introduce another sub precosheaf of cdgas $\calO_{md}(\X)$ of $\calO(\X)$, which is defined by the condition that covector component of the de Rham differential of such a function is mollified. Formally, we define $\calO_{md}(\X)$ by the pullback square:
\[\xymatrix{
\calO_{md}(\X)  \ar[r] \ar[d] & \widehat \Sym^\bullet(\scg^*[-1])\otimes \scg^!_c[-1] \ar[d]  \\
 \calO(\X) \ar[r]^{{d_\dR}}   & \bar \Gamma_c(\X,\LL_\X) 
}\]
Note that this subspace is closed under precomposition with the polydifferential operators defining the Chevalley--Eilenberg differential, so is indeed a sub cdga.

Finally, we define the total space of a vector bundle:
\begin{defn} Let $\X$ be a local moduli problem with corresponding local $L_\infty$-algebra $\scg=\T_o[-1]\X$, and let $\V$ a vector bundle on $\X$ corresponding to a local $L_\infty$ module $\scv$ for $\scg$. We define the \emph{total space} $|\V|$ as the local moduli problem defined by
$$\T_o[-1]|\V|=\scg\ltimes \scv[-1].$$
\end{defn}
Note that in particular, we have an isomorphism of precosheaves of $\calO(\X)$ dg-modules
$$\calO(|\V|)=\bar \Gamma_c(\X,\widehat\Sym^\bullet(\V^!)).$$

\begin{eg} Let $\X$ be a formal moduli problem. We define the local moduli problem $T^*[n]\X:=|\LL_\X[n]|$ as the total space of the cotangent bundle of $\X$. We will see in Subsection $\ref{slss}$ that such formal moduli problems admit a natural constant coefficient, strictly local $n$-shifted symplectic structure.
\end{eg}

\subsection{Chevalley--Eilenberg Differentials and Local Vector Fields}\label{lvf}

In this subsection, we introduce a notion of local sections of a vector bundle over a local moduli problem. The primary motivation for this notion is Proposition $\ref{dtla}$, which states that the structure maps of a local $L_\infty$ algebra $\scg=\T_o[-1]\X$ can be interpreted as cohomological local vector fields on the local moduli problem $\tilde \X$ corresponding to $\scg$ equipped with the trivial local $L_\infty$ structure. In the next subsection, we will also use this description together with the notion of strictly local $(-1)$-shifted symplectic structure to explain how local action functionals give rise to local $L_\infty$ algebras describing the spaces of solutions to the Euler--Lagrange equations corresponding to the actions.

In general, we define the notion of local sections of a vector bundle on a local moduli problem as follows:
\begin{defn}
Let $\scg$ a local $L_\infty$ algebra and $\scv$ a $\scg$-module. We define the \emph{local Chevalley--Eilenberg cochains $C^\bullet_\loc(\scg;\scv)$ on $\scg$ with coefficients in $\scv$} by
$$C^\bullet_\loc(\scg;\scv)=\prod_{n\geq 0} \Poly(\scg[1]^{\otimes n},\scv)_{S_n}$$
equipped with the Chevalley--Eilenberg differential.

Let $\X$ a local moduli problem with $\T_o[-1]\X=\scg$ and $\V$ a vector bundle over $\X$ corresponding to the $\scg$-module $\scv$. The \emph{sheaf $\Gamma_\loc(\X,\V)$ of spaces of local sections of $\V$ over $\X$} is defined as
$$\Gamma_\loc(\X,\V)=C^\bullet_\loc(\scg;\scv).$$
\end{defn}
Note that there is a natural inclusion $\Gamma_\loc(\X,\V)\into \Gamma(\X,\V)$, which is closed under the Chevalley--Eilenberg differential on the latter, since compositions of polydifferential operators are again polydifferential operators; thus, the Chevalley--Eilenberg differential above is indeed well-defined. Further, note that spaces of polydifferential operators are by definition given by sheaves of $C^\infty_M$ modules.

We now restrict our attention to the spaces of local vector fields. Under the sequence of inclusions
$$\Gamma_\loc(\X,\T_\X[k])\into \Gamma(\X,\T_\X[k])\into \Der^k(\calO(\X)),$$
the space of local, cohomological degree $k$ vector fields is closed under the $k$-shifted Lie algebra structure on $\Der^k(\calO(\X)_U)$, that is, the Lie algebra structure on $\Der^k(\calO(\X)_U)[-k]$, given by commutator of derivations; this also follows from the fact that compositions of polydifferential operators are again polydifferential operators. In particular, with the Chevalley--Eilenberg differential given above, this implies that $\Gamma_\loc(\X,\T_\X[k])$ defines a sheaf of $k$-shifted dg Lie algebras.

We now have the following description of local $L_\infty$ algebra structures:
\begin{prop}\label{dtla}
Let $L$ be a smooth vector bundle on $M$ with sheaf of sections $\scg$, and let $\tilde \X$ be the sheaf of affine spaces corresponding to $\scg$ thought of as a trivial local $L_\infty$ algebra. The following are equivalent:
\begin{itemize}
\item A collection of polydifferential operators
$$\{l_n:\scg^{\otimes n}\to \scg[2-n]\}_{n\in\N}$$
making $\scg$ into a local $L_\infty$ algebra.
\item A vector field
$$Q_\scg\in \Gamma_{\loc}(\tilde \X,\T_{\tilde \X}[1])_M$$
with vanishing polynomial degree $0$ term, satisfying $[Q_\scg,Q_\scg]=0$. 
\end{itemize}
\end{prop}

\subsection{Strictly Local Symplectic Structures and Classical Field Theories}\label{slss}

In this subsection we define the notion of a strictly local $n$-shifted symplectic structure $\omega$ on local moduli problem $\X$, and define a (non-degenerate) classical field theory as a local moduli problem equipped with such a structure. Further, we recall the main result of \cite{CostelloGwilliam} on classical field theories, which states that the observables of a classical field theory form a $P_0$ factorization algebra on the underlying manifold.

\begin{defn}
Let $\X$ be a local moduli problem with $\T_o[-1]\X=\scg$ the corresponding local $L_\infty$ algebra with underlying vector bundle $L$. A \emph{strictly local $n$-shifted symplectic structure on $\X$} is an element
$$\omega\in \scg^*\otimes \scg^![n-2] \subset \bar\Gamma (\X,\wedge^2 \LL_\X[n])$$
defined by a bundle map $\omega:L \to L^![n-2]$ which is an isomorphism on each fibre, symmetric in $L$, and satisfies
$$ \calL_{Q_\scg} \omega=0,$$
where $Q_\scg\in \Gamma_\loc(\tilde \X,\T_{\tilde \X}[1])_M$ is the cohomological local vector field defining the $L_\infty$ structure on $\scg$.
\end{defn}

We let $\Pi_\omega:L^!\to L[2-n]$ denote the inverse of $\omega$, which is interpreted as a non-degenerate shifted Poisson tensor on $\X$; we will see that such $\Pi_\omega$ can be interpreted as a closed element of an appropriate space of bivector fields. The equation $\calL_{Q_\scg}\omega=0$ is simply the statement that $Q_\scg$ is a symplectic vector field for $\omega$, which we interpret equivalently in terms of $\Pi_\omega$, and for which the precise definition can be given from Proposition $\ref{SB}$.

\begin{rmk}
This definition deserves a few unpacking remarks.
\begin{itemize}
\item We use the term strictly local to emphasize that $\omega$ and correspondingly $\Pi_\omega$ are required to be bundle maps; when we generalize this notion to that of local $(-1)$-shifted Poisson structures, we will allow for arbitrary differential operators between the sheaves of sections of these bundles.
\item Note that in the above definition $\omega$ is constant coefficient as a section over $\X$, and thus is homotopically strictly symplectic, in the sense that it is closed on the nose for the de Rham differential. It is shown in \cite{CostelloGwilliam} that the space of $n$-shifted symplectic structures on a formal moduli problem is equivalent to the space of those which are constant coefficient, so that we can work with this as a model for all symplectic structures without loss of generality for our purposes here; thus, we have defined a general strictly local shifted symplectic structure as such. The corresponding result is not true in the Poisson case, and accordingly our definition of local $(-1)$-shifted Poisson structures will not require being constant or homotopically strict.
\item Being constant coefficient, the form $\omega$ is closed if and only if it is closed under the internal differential, which is realized by the condition $\calL_{Q_\scg}\omega=0$ in the formal setting.
\end{itemize}
\end{rmk}

\begin{eg} Let $M=\{*\}$ be a point and $\X$ a local moduli problem over $M$, or equivalently a formal moduli problem $X$. Then the differential geometric aspects of the definition hold vacuously, and thus the definition of strictly local $n$-shifted symplectic structure on $\X$ reduces to the definition of a strict, constant coefficient $n$-shifted symplectic structure on the formal moduli problem $X$. In the formal case, such a structure $\omega \in \Gamma(X,\wedge^2 \LL_X[n])$ is simply a symmetric, non-degenerate linear map $\omega: \g \to \g^*[n-2]$ satisfying $\calL_{Q_\g}\omega=0$.
\end{eg}

\begin{eg}\label{csf} Let $\X$ be a formal moduli problem and $T^*[n]\X=|\LL_\X[n]|$ be the total space of the $n$-shifted cotangent bundle to $\X$. Then $T^*[n]\X = |\LL_\X[n]|$ has a canonical strictly local $n$-shifted symplectic structure, defined by the pairing $\langle \ ,\ \rangle_{\LL_\X}$ in the usual way.
\end{eg}

\begin{eg}
Let $X$ be a formal moduli problem over $k$ with corresponding $L_\infty$ algebra $\g=\T_o[-1]X$ finite dimensional as a vector space. Further, let $\eta\in \Gamma(X,\wedge^2 \LL_X[n])$ be a strict, constant coefficient, $n$-shifted symplectic structure on $X$. Now, let $M$ be a $d$-dimensional smooth, oriented manifold: from the orientation there is a natural isomorphism corresponding to integration $I_M:\wedge^\bullet T^*_M \to (\wedge^\bullet T^*_M)^![-d]$, and moreover
$$\omega:=I_M\otimes \eta: \wedge^\bullet T^*_M \otimes \g \to (\wedge^\bullet T^*_M \otimes \g)^![n-d-2].$$
defines a strictly local $(n-d)$-shifted symplectic structure on the local moduli problem $\widehat{\underline{\mr{Map}}}_o(M_\dR,B\g)$ over $M$. This is the AKSZ symplectic structure defined in $\cite{AKSZ}$.

More generally, let $\Y$ be a formal moduli problem over a manifold $N$ with $\scg=\T_o[-1]\Y$ the corresponding local $L_\infty$ algebra with underlying vector bundle $L$, and let $\eta$ be a strictly local $n$-shifted symplectic structure on $\Y$. Then
$$\omega:=I_M\boxtimes \eta: \wedge^\bullet T^*_M \boxtimes L \to  (\wedge^\bullet T^*_M \boxtimes L)^![n-d-2]$$
defines a strictly local $(n-d)$-shifted symplectic structure on the local moduli problem $\widehat{\mr{Map}}_o(M_\dR,\Y)$ over $M\times N$.
\end{eg}

Given a strictly local $n$-shifted symplectic structure on $\X$, we obtain an isomorphism of families of $\calO(\X)$ dg-modules between the spaces of mollified sections
$$\omega:\Gamma(\X,\T_\X)\leftrightarrows \Gamma(\X,\LL_\X[n]):\Pi_\omega$$
and analogously between the spaces of general and local sections. Note that this claim for local sections uses crucially that there is a locality condition on $\omega$.

We now make the main definition of this section:
\begin{defn}
A \emph{classical field theory} on $M$ is a local moduli problem $\X$ over $M$ equipped with a strictly local $(-1)$-shifted symplectic structure $\omega$ on $\X$.
\end{defn}

\begin{rmk}
Note that for quantizing classical field theories in the formalism of Costello, one also needs to require an ellipticity hypothesis on the differential operators defining the local moduli problem; we will not discuss this subtlety here. Further, in the work of Costello and Gwilliam, in order to ensure that $\calO(\X)$ and $\calO_{md}(\X)$ define homotopy equivalent commutative factorization algebras, and to homotopy transfer the $P_0$ structure on $\calO_{md}(\X)$, defined in the remainder of this subsection, to $\calO(\X)$, one needs an ellipticity hypothesis to employ the key lemma of Atiyah and Bott $\cite{AtiyahBott}$; although outside of this situation there is some uncertainty about which are the philosophically correct spaces of observables, we state our results without the ellipticity hypothesis.
\end{rmk}

The motivation for this definition is the equivalence in Proposition $\ref{slla}$ given in the following subsection, which states that classical field theories in this sense are equivalent to the presheaves of formal moduli spaces of solutions to the Euler--Lagrange equations of a local action functional.

Further, the precosheaves of spaces of functions $\calO_{md}(\X)$ corresponding to a classical field theory inherit a non-degenerate Poisson bracket from the local symplectic form. The spaces of functions $\calO_{md}(\X)_U$ are interpreted as the observables of the classical field theory supported on the open set $U\subset M$. The main result about the algebraic structure of these spaces of observables is:

\begin{prop} Let $\X$ be a local moduli problem over $M$ and $\omega$ a strictly local $(-1)$-shifted symplectic structure on $\X$. Then $\calO_{md}(\X)$ defines a $P_0$ factorization algebra on $M$, with Poisson bracket
$$\{\cdot,\cdot\}:\calO_{md}(\X)^{\otimes 2}\to \calO_{md}(\X)\quad\quad\text{defined by}\quad\quad \{f,g\}_U=\langle \Pi_\omega(d^\textup{dR}f), d^\textup{dR}g \rangle_{\T_\X}$$
for $f,g\in\calO_{md}(\X)_U$.
\end{prop}

In particular, this implies that the Poisson bracket on $\calO_{md}(\X)$ satisfies a physical locality condition which is crucial for quantization: for $U,V\subset W$ disjoint open subsets of $M$, the cosheaf $\calO_{md}(\X)$ of $P_0$ algebras satisfies
$$\{ \iota_U f,\iota_V g\}_W=0$$
for arbitrary $f\in \calO_{md}(\X)_U,g\in \calO_{md}(\X)_V$, where $\iota_U:\calO_{md}(\X)_U\to \calO_{md}(\X)_W$, and similarly for $V$, are the cosheaf structure maps.

\subsection{Local Action Functionals \label{laf}}

In this section we recall the equivalence between classical field theories as defined above and local action functionals satisfying the classical master equation, which are the usual defining data for a classical field theory in the BV formalism.

Let $E$ be a smooth graded vector bundle on $M$ with sheaf of sections $\calE$, which we think of as the space of fields in the BV formalism, in particular including the anti-fields, and let $\calJ(\calE)$ denote the sheaf of sections of the infinite jet bundle $J(E)$ of $E$. Recall that there is a natural flat connection on $J(E)$ making $\calJ(\calE)$ into a sheaf of left modules for the sheaf of differential operators $\calD_M$ on $M$, and that the sheaf of densities $\Dens_M$ on $M$ is naturally a right $\calD_M$ module.

Further, letting $L=E[-1]$ and $\scg$ the sheaf of sections of $L$, recall that the formal moduli problem $\tilde \X$ corresponding to $\scg$, considered as a trivial $L_\infty$ algebra, describes $\calE$ as a sheaf of affine spaces, on which local action functionals should define functions.

The sheaf of local action functionals on $\calE$ is defined by
$$\calO_\loc(\tilde \X)= \Dens_M \otimes_{\calD_M} \widehat \Sym^\bullet_{C^\infty_M}(\calJ(\calE)^\vee)$$
where $(\cdot)^\vee$ denotes dual in the category of sheaves of $C^\infty_M$ modules. Note that $\widehat \Sym^\bullet_{C^\infty_M}(\calJ(\calE)^\vee)$ is the sheaf of sections of a vector bundle whose fibre over each point $x\in M$ is given by the space of formal power series in the values of the field and its derivatives at the point $x$; this is essentially the usual space of Lagrangian functions without a fixed density to integrate against. Taking the tensor product over $\calD_M$ with the sheaf of densities gives the space of Lagrangian densities modulo the natural relations imposed from integration by parts, which is precisely the desired space of local action functionals. 

Note that elements of $\calO_\loc(\tilde \X)$ do not precisely define functions on $\calE_U$ for each $U\subset M$ because there is no condition to guarantee the relevant integral converges. However, the Lagrangian densities do define functions on the spaces of compactly supported sections, and as such, their de Rham differentials can be understood as vector fields. In fact, we have:

\begin{prop}\label{ldrd} There exists a natural de Rham differential map
$$d_\textup{dR}:\calO_\loc(\tilde \X)\to \Gamma_\loc(\tilde \X, \LL_{\tilde \X})$$
\end{prop}
\begin{proof}
Taking the formal adjoint of a tensor factor $\calJ(\calE)^\vee\cong \Diff(\calE,C^\infty_M)$ yields an element of $\Diff(\Dens_M,\calE^!)$, to which we apply the isomorphism $\Dens_M\otimes_{\calD_M} \Diff(\Dens_M,\calE^!)\cong \calE^!$. Taking a symmetrizing sum over the tensor factors, this identifies the space of polynomial degree $j$ local action functionals modulo constants with the space of polynomial degree $j-1$ local $1$-forms.
\end{proof}

Now, given a strictly local $(-1)$-shifted symplectic pairing $\omega$ on $\tilde \X$, corresponding physically to the duality pairing on $\calE$ between fields and anti-fields, we in particular obtain an isomorphism
$$\Pi_\omega:\Gamma_\loc(\tilde \X, \LL_{\tilde \X})\to\Gamma_\loc(\tilde \X, \T_{\tilde \X}[1]).$$
From this, we define the Hamiltonian vector field corresponding to $S\in \calO_\loc(\tilde \X)$ as
$$Q_S=\Pi_\omega\circ d_\dR (S).$$
From Proposition $\ref{dtla}$, we have that such a vector field is equivalent to a local $L_\infty$ algebra structure on $\scg:=\calE[-1]$ given that it satisfies $Q_S^2=0$. Toward stating an equivalent condition on action functionals to ensure their Hamiltonian vector fields are Hamiltonian, we would like to introduce a $P_0$ bracket on the space of local action functionals. However, the space of local action functionals is not closed under multiplication as functions, so we only obtain a Lie algebra structure:

\begin{prop} There exists a differentiation pairing
$$\Gamma_\loc(\tilde \X, \T_{\tilde \X}[1])\otimes \calO_\loc(\tilde \X)\to  \calO_\loc(\tilde \X)[1]\quad\quad\text{denoted by}\quad\quad Q\otimes f \to Q(f)$$
such that the bracket $\{\cdot,\cdot\}_\omega: \calO_\loc(\tilde \X)^{\otimes 2}\to\calO_\loc(\tilde \X)[1]$ defined by
$$\{f,g\}_\omega=Q_f(g)$$
defines a $1$-shifted Lie algebra structure on $\calO_\loc(\tilde \X)$. Moreover, the map
$$\Pi_\omega\circ d_\textup{dR}: \calO_\loc(\tilde \X)\to \Gamma_\loc(\tilde \X,\T_{\tilde \X}[1])$$
is a map of $1$-shifted Lie algebras.
\end{prop}

In terms of this bracket, we have
$$ Q_S^2(f)= \{S,\{S,f\}_\omega\}_\omega = \frac{1}{2}\{\{S,S\}_\omega,f\}_\omega$$
which vanishes for all $U\subset M$ and $f\in\calO(\X)_U$ if and only if $\{S,S\}_\omega=0$. Thus, we have

\begin{prop}\label{slla} Let $L$ be a smooth vector bundle on $M$ with sheaf of sections $\scg$, $\tilde \X$ be the sheaf of affine spaces corresponding to $\scg$ thought of as a trivial local $L_\infty$ algebra, and $\omega$ a strictly local $(-1)$-shifted symplectic structure on $\tilde \X$. The following are equivalent:
\begin{itemize}
\item A collection of polydifferential operators
$$\{l_n:\scg^{\hat\otimes n}\to \scg[2-n]\}_{n\in\N}$$
making $\scg$ into a local $L_\infty$ algebra such that $\omega$ defines a strictly local symplectic structure on the local moduli problem $\X$ corresponding to $\scg$.
\item A vector field
$$Q_\scg\in \Gamma_{\loc}(\tilde \X,\T_{\tilde \X}[1])_M$$
with vanishing polynomial degree $0$ term, and satisfying $[Q_\scg,Q_\scg]=0$ and $\calL_{Q_\scg}\omega=0$.

\item An action functional $S\in\calO_\loc(\tilde \X)_M$, with vanishing polynomial degree $0$ and $1$ terms, and satisfying $\{S,S\}_\omega=0$.
\end{itemize}
\end{prop}

Again, the precise definitions needed to make sense of the expression $\calL_{Q_\scg}\omega$ are given in terms of the corresponding Poisson tensor $\Pi_\omega$ in Proposition $\ref{SB}$ below.

\begin{rmk}\label{sslla}
Note that there is a natural extension of this proposition allowing for an arbitrary shift of cohomological degree of the polydifferential operators, vector fields and action functionals which are shown to be equivalent.
\end{rmk}

Motivated by the above proposition, we recover the classical definition:
\begin{defn} An action functional $S\in\calO_\loc(\tilde \X)$ is said to satisfy the \emph{classical master equation} if
$$\{S,S\}_\omega=0.$$
\end{defn}

\subsection{Local Homotopy $(-1)$-Shifted Poisson Structures and Degenerate Classical Field Theories \label{dcft}}

This subsection marks the beginning of the original work of this paper. We formulate the definition of spaces of multilocal $(-1)$-shifted polyvector fields, describe their basic geometric properties, and in terms of these give the key definitions of general local $(-1)$-shifted Poisson structures on local moduli problems and degenerate classical field theories. Further, this section includes the first main theorem of this paper, which states that the observables of a degenerate classical field theory yield a $P_0$ factorization algebra on the underlying space.

\begin{defn} Let $\X$ be a local moduli problem and $\scg=\T_o[-1]\X$ the corresponding local $L_\infty$ algebra.

Define the \emph{sheaf of $(-1)$-shifted, multilocal $j$-polyvector fields} by
$$\Gamma_{m\loc}(\X,\Sym^j \T_\X) = \prod_{k\in\N} \Poly(\scg[1]^{\otimes k}\otimes \scg^![-1]^{\otimes (j-1)},\scg[1])_{S_k\times S_j}$$
equipped with the Chevalley--Eilenberg differential defined below, where the subscript $S_k\times S_j$ denotes the coinvariants with respect to the natural symmetric group actions.
\end{defn}

We denote the factors in the direct product defining $\Gamma_{m\loc}(\X,\Sym^j \T_\X)$ by $\Gamma_{m\loc}^k(\X,\Sym^j \T_\X)$. These are the polynomial degree $k$ components of the space of $(-1)$-shifted multilocal $j$-polyvector fields on $\X$: informally, thinking of $\scg^!$ as dual to $\scg$ and applying hom-tensor adjunction, an element of $\Gamma_{m\loc}^k(\X,\Sym^j \T_\X)$ is a degree $k$ polynomial function of $\scg[1]$ valued in $\Sym^j(\scg[1])$. Note that these spaces define sheaves on $M$, as with the local spaces of sections defined previously.

We denote the space of all $(-1)$-shifted, multilocal polyvector fields, and its completion, by
$$\Gamma_{m\loc}(\X,\Sym^\bullet \T_\X)=  \bigoplus_{j\geq 0} \Gamma_{m\loc}(\X,\Sym^j \T_\X) \quad\quad\text{and}\quad\quad \Gamma_{m\loc}(\X,\widehat {\Sym}^{\bullet}\T_\X)=\prod_{j\geq 0}\Gamma_{m\loc}(\X,\Sym^j \T_\X).$$

In order to define the notion of multilocal Poisson bivector, and to introduce the Chevalley--Eilenberg differential on all the spaces of multilocal polyvector fields, we will need to discuss the Schouten bracket of multilocal vector fields.

\begin{prop}\label{SB} There exists a map, called the \emph{Schouten bracket}, of sheaves of vector spaces on $M$
$$[\cdot,\cdot]:\Gamma_{m\loc}(\X,\Sym^j\T_\X)\otimes\Gamma_{m\loc}(\X,\Sym^l\T_\X)\to \Gamma_{m\loc}(\X,\Sym^{j+l-1}\T_\X)$$
making $\Gamma_{m\loc}(\X,\Sym^{\bullet}\T_\X)$ a sheaf of graded Lie algebras.
\end{prop}
\begin{proof} The usual algebraic definition of the Schouten bracket is well-defined, as a sum of compositions of polydifferential operators.
\end{proof}

Note that, as with usual polyvector fields, $\Gamma_{m\loc}(\X,\T_\X)$ forms a sub Lie algebra of $\Gamma_{m\loc}(\X,\Sym^\bullet\T_\X)$ which acts on $\Gamma_{m\loc}(\X,\Sym^j\T_\X)$ for each $j\geq 0$.
Further, by definition there is an identification of sheaves $C^\infty_M$ modules
$$\Gamma_{m\loc}(\X,\T_\X)= \Gamma_{\loc}(\X,\T_\X).$$
The above identification is also an identification of sheaves of Lie algebras, where the Lie algebra structure on $\Gamma_{m\loc}(\X,\T_\X)$ is given by the Schouten bracket above, and the Lie algebra structure on $\Gamma_\loc(\X,\T_\X)$ is as discussed in Section $\ref{lvf}$; the definition of the Schouten bracket given above, when restricted to vector fields, is precisely the same composition of polydifferential operators defining the Lie algebra structure on local vector fields.

In particular, the image under this identification of the vector field $Q_\scg\in\Gamma_{\loc}(\X,\T_\X[1])$ defining the local $L_\infty$ structure on $\scg$ defines by adjoint action a cohomological degree $1$ differential on $\Gamma_{m\loc}(\X,\Sym^k\T_\X)$ for each $k\geq 0$, which we define to be the Chevalley--Eilenberg differential. Note that the identification in the preceding proposition is an isomorphism of sheaves of dg Lie algebras with respect to the Chevalley--Eilenberg differentials as defined.

In keeping with the usual interpretation of the Schouten bracket, we use the notation
$$\calL_Q X:=[Q,X]$$
for $Q\in\Gamma_{m\loc}(\X,\T_\X)$ and $X\in\Gamma_{m\loc}(\X,\Sym^k\T_\X)$.

We can now define the desired spaces of $(-1)$-shifted (homotopy) Poisson structures on local moduli problems:
\begin{defn} Let $\X$ be a local moduli problem with $\scg=\T_o[-1]\X$ the corresponding local $L_\infty$ algebra. A \emph{local $(-1)$-shifted Poisson structure} on $\X$ is a cohomological degree $1$ multilocal bivector field
$$\Pi \in \Gamma_{m\loc}(\X,\Sym^2 \T_\X)[1]$$
such that
$$\calL_{Q_\scg}\Pi=0\quad\quad\text{and}\quad\quad [\Pi,\Pi]=0.$$

More generally, a \emph{local $(-1)$-shifted homotopy Poisson structure} on $\X$ is a cohomological degree $1$, non-homogeneous, multilocal polyvector field
$$\bar\Pi=\{\Pi_j\}_{j=2}^\infty \in \Gamma_{m\loc}(\X,\widehat{\Sym}^{\bullet}\T_\X)[1]$$
satisfying the equations
$$[Q_\scg+\bar\Pi, Q_\scg+\bar\Pi]=0$$
defining a $(-1)$-shifted homotopy Poisson structure.
\end{defn}
This notion strongly generalizes the notion of strictly local $(-1)$-shifted symplectic structure on a local moduli problem $\X$ and thus the notion of classical field theory, and will allow us to consider many physically and mathematically interesting examples which do not satisfy the previously given definition of classical field theory.

\begin{eg} Let $M=\{*\}$ be a point. Then the differential geometric subtleties of the preceding definition hold vacuously, and it reduces to the usual definition of a $(-1)$-shifted (homotopy) Poisson structure on a formal moduli problem.
\end{eg}

\begin{eg} Note that the space of constant coefficient local $(-1)$-shifted Poisson structures is identified with a subspace of the space of differential operators $\Diff(\scg^!,\scg[3])$. In particular, given a strictly local $(-1)$-shifted symplectic structure $\omega:L\to L^![-3]$, its corresponding Poisson tensor, defined as the inverse bundle map $\Pi_\omega:L^!\to L[3]$, defines such a differential operator.
\end{eg}

The boundary theory construction, given in Theorem $\ref{BTC}$, will provide many more interesting examples of local $(-1)$-shifted homotopy Poisson structures. Several explicit examples of applications of this theorem occur in Section $\ref{section:examples}$.

Our main theorem about local $(-1)$-shifted Poisson structures is the following:
\begin{theo}\label{dftf} Let $\X$ be a local moduli problem over $M$, equipped with a local $(-1)$-shifted Poisson structure $\Pi\in \Gamma_{m\loc}(\X,\Sym^2 \T_\X)[1]$. Then $\calO_{md}(\X)$ defines a $P_0$ factorization algebra on $M$.
\end{theo}

Motivated by this, we make the following definition:

\begin{defn}
A \emph{degenerate classical field theory} is a local moduli problem $\X$ together with a local $(-1)$-shifted Poisson structure $\Pi\in \Gamma_{m\loc}(\X,\Sym^2\T_\X)[1]$ on $\X$.
\end{defn}

We will discuss many examples of degenerate classical field theories, and see that they provide models at the classical level of many interesting quantum field theories which do not admit Lagrangian descriptions.

\newpage

\section{Classical Boundary Theories and Universal Bulk Theories}\label{bt}

In this section, we introduce the notions of phase spaces and boundary conditions for classical field theories on manifolds with boundary, in the language outlined in the previous section. Further, given a fixed classical field theory on a manifold with boundary, together with a choice of boundary condition for it, we construct the corresponding boundary theory as a degenerate classical field theory on the boundary, as described in the introduction. Finally, we construct the universal bulk theory of a given degenerate classical field theory.

\subsection{A Model for the $(n-1)$-Shifted Homotopy Poisson Structure on Lagrangians in $n$-Shifted Symplectic Formal Moduli Problems}\label{TM}

In this subsection we explain the existence of, and give an explicit model for computing, the $(n-1)$-shifted homotopy Poisson structure on a derived Lagrangian in an $n$-shifted symplectic formal moduli problem. As far as we understand, Costello and Rozenblyum are the first who proved a version of the result in the global algebraic setting \cite{Rozenblyum}. Since we need to prove an analogue of this result in the setting of local moduli problems, where we insist on having strict models for the maps involved, we only provide this construction in an analogously strict sense, and do not discuss its homotopical invariance. Closely related constructions to this one have also appeared in \cite{MelaniSafronov}, \cite{Theo}.

Let $X$ be a formal moduli problem over $k$ with corresponding $L_\infty$ algebra $\g=\T_o[-1]X$ finite dimensional as a vector space and denote the $L_\infty$ brackets of $\g$ by $\{l_n\}_{n\in\N}$. Further, let $\omega\in \Gamma(X,\wedge^2 \LL_X[n])$ a strict, constant coefficient, $n$-shifted symplectic structure on $X$, which is determined by a symmetric linear map
$$\omega: \g \to \g^*[n-2]$$
which is an isomorphism of vector spaces satisfying an invariance condition in terms of the $L_\infty$ brackets $\{l_n\}$; this condition can be described as requiring $\calL_{Q_\g}\omega=0$, for $Q_\g\in \Gamma(X,\T_X[1])$ the degree $1$ cohomological vector field corresponding to the Chevalley--Eilenberg differential on $C^\bullet(\g;k)$. Recall again that the space of such strict, constant coefficient $n$-shifted symplectic structures is equivalent to the space of all $n$-shifted symplectic structures \cite{CostelloGwilliam}, so that this simplification is without loss of generality for our purposes here.

The condition $\calL_{Q_\g}\omega=0$ can also be interpreted as requiring that the vector field $Q_\g$ be symplectic with respect to $\omega$, and in the formal setting any such vector field is a Hamiltonian vector field corresponding to a cohomological degree $n+1$ function $S\in \calO(X)[n+1]$. Further, the condition $Q_\g^2=0$ is equivalent to the condition $\{S,S\}_\omega=0$; this is the finite dimensional toy model of the result stated in Proposition $\ref{slla}$, although with a cohomological grading shift, as discussed in Remark $\ref{sslla}$.

Let $L_+$ be another formal moduli problem over $k$ with corresponding $L_\infty$ algebra $\el_+=\T_o[-1]L_+$ finite dimensional and denote the $L_\infty$ brackets of $\el_+$ by $\{l^+_n\}_{n\in\N}$. Further, let $F:L_+\to X$ be a map of formal moduli problems such that the corresponding $L_\infty$ map $f:\el_+\to \g$ is strict, and is injective as a map of the underlying vector spaces. Under these conditions, the map $L_+\to X$ is Lagrangian if and only if the induced map $\el_+ \to \el_+^*[n-2]$ is zero. We will call such Lagrangians as \emph{regular embedded}.

\begin{prop} Let $X$ be a formal moduli problem equipped with an $n$-shifted symplectic structure $\omega$, and $L_+\into X$ a regular embedded Lagrangian. Then there exists a natural $(n-1)$-shifted homotopy Poisson structure $\bar \Pi\in \Gamma(L_+, \widehat\Sym^\bullet (\T_{L_+}[-n] ) )[n+1]$ on $L_+$.
\end{prop}
\begin{proof}
Since $L_+\into X$ is a regular embedding, the normal cone $L_+\into C_X L_+$ is isomorphic to the total space of the normal bundle $L_+\into |N_{L_+/X}|$, which is isomorphic as a vector bundle to $\LL_{L_+}[n]$, the total space of which we denote by $T^*[n]L_+$. The formal deformation of the normal cone construction yields a deformation of $T^*[n]L_+$ to $X$, given by a vector field $Q_d\in \Gamma(T^*[n]L_+,\T_{T^*[n]L_+}[1])$ such that $Q_\g=Q_{\T_o(T^*[n]L_+)}+Q_d$. The vector field $Q_d$ is equivalent, via the $n$-shifted symplectic structure on $T^*[n]L_+$, to an action functional
$$S_d\in \calO(T^*[n]L_+)[n+1]\cong \Gamma(L_+,\widehat\Sym^\bullet ( \T_{L_+}[-n]) )[n+1]$$
Defining $\Pi= S_d$, the classical master equation for the action functional $S\in \calO(X)[n+1]$ then implies
$$[Q_{\el_+}+\Pi, Q_{\el_+}+\Pi]=0$$
so that $\Pi$ defines an $(n-1)$-shifted homotopy Poisson structure on the formal moduli problem $L_+$.
\end{proof}

We now give an explicit presentation of the $(n-1)$-shifted homotopy Poisson structure given above: choose a complementary vector subspace $\el_-\into \g$ such that the induced map $\el_-\to \el_-^*[n-2]$ also vanishes; the following calculation will only depend on this choice up to homotopy. Now, the vector space isomorphism $\omega:\g\to \g^*[n-2]$ induces an isomorphism
$$\omega|_{\el_-}:\el_-\to \el_+^*[n-2]$$
so that we have an isomorphism of vector spaces $\g\cong \el_+\oplus \el_+^*[n-2]$. In terms of this isomorphism, we have the isomorphism of vector spaces
$$ \calO(X)[n+1] = \widehat\Sym^\bullet( \g^*[-1])[n+1] \cong \widehat\Sym^\bullet(\el_+^*[-1])\otimes_k \widehat\Sym^\bullet(\el_+[1-n])[n+1] $$
and thus we can decompose $S\in \calO(X)[n+1]$ as
$$
S = \sum_{j\geq 0} S_j \quad\quad\text{where}\quad\quad S_j \in \widehat \Sym^\bullet(\el_+^*[-1])\otimes_k \Sym^j(\el_+[1-n])[n+1] \cong \Gamma(L_+, \Sym^j( \T_{L_+}[-n]))[n+1]
$$
so that we obtain a sequence of elements of the $(n+1)^{th}$ shift of the $(n-1)$-shifted polyvector fields on $L_+$. Note that in particular
$$ S_1 \in \Gamma(L_+,\T_{L_+})[1] \quad\quad\text{and}\quad\quad S_2 \in \Gamma(L_+,  \T_{L_+}^{\otimes 2} )[1-n] ,$$ where we have symmetricity or anti-symmetricity of tensor powers depending on the parity of $n$, although these are not a priori closed elements for the cohomological differentials on these complexes. Now, we have:

\begin{prop} Let $X$ be a formal moduli problem equipped with an $n$-shifted symplectic structure $\omega$, and with $L_\infty$ structure on $\g=\T_o[-1]X$ given by an action functional $S\in \calO(X)[n+1]$. Further, let $L_+\into X$ be a regular embedded Lagrangian, and fix a complementary vector space $\el_-$ to $\el_+\into \g$. Then $S_0=0$, $S_1=Q_\g \in \Gamma(X,\T_X[1])$ the cohomological vector field defining the $L_\infty$ structure on $\g$,  and
$$\bar\Pi = \{S_j\}_{j=2}^\infty \quad \in \ \prod_{j\geq 2}\Gamma(L_+, \Sym^j( \T_{L}[-n]))[n+1]$$
defines the $(n-1)$-shifted homotopy Poisson structure on $L_+$.
\end{prop}
\begin{proof}
Recall that $S\in\calO(X)[n+1]$ is given by
$$S(x)=\sum_{n\geq 0} \omega(x,l_n(x^{\otimes n}))$$
for each $x\in\g[1]$, where $\{l_n:\g^{\otimes n}\to \g[2-n]\}_{n\geq 0}$ are the $L_\infty$ brackets on $\g$. Since $\el_+$ is a Lagrangian subspace of $\g$ which is closed under $\{l_n\}_{n\geq 0}$, it is clear that the above expression for $S$ vanishes on tensor products of elements $l\in \el_+$, so that $S_0=0$. 

Interpreting the Lagrangian condition as the equality of $\omega$ with its restriction $\omega|_{\el_-}:\el_-\to \el_+^*[n-2]$, the $S_1$ component of the above expression for $S$ is identified with the vector field
$$S_1(x)=\sum_{n\geq 0} l_n(x^{\otimes n})$$
which is precisely $Q_\g$, the vector field defining the $L_\infty$ brackets on $\g$, as claimed.

Finally, under the above identifications, the $n$-shifted classical master equation on $S$ is equivalent to the equation
$$[Q_\g+\bar\Pi,Q_\g+\bar\Pi]=0$$
defining an $(n-1)$-shifted homotopy Poisson structure on $X$.
\end{proof}

An alternative way of understanding this construction is via the notion of higher Poisson centre, which comes from thinking of a Lagrangian as a coisotropic rather than isotropic subspace. This notion has appeared in \cite{Safronov} in the affine case; we will define it here for formal moduli problems and later give the analogous construction for local moduli problems.

Let $L$ be a formal moduli problem with $\el=\T_o[-1]L$ the corresponding $L_\infty$-algebra and let
$$\bar\Pi=\{\Pi_j\}_{j=2}^\infty\quad \in \  \prod_{j\geq 2}\Gamma(L, \Sym^j( \T_{L}[-n]))[n+1]$$
not neccesarily Poisson.
Recall that the $L_\infty$ $\el$-module corresponding to the vector bundle $\LL_L[n]$ on $L$ is $\el^*[n-1]$ so that we have an isomorphism of $P_{n+1}$ algebras
$$\calO(T^*[n]L) = C^\bullet(\el\ltimes \el^*[n-2];k) \cong \Gamma(L, \widehat\Sym^\bullet(\T_L[-n])) $$
and the latter space is the completion of the space of $(n-1)$-shifted polyvector fields of all degrees, equipped with the Schouten bracket, which is of cohomological degree $-n$. Under this isomorphism, $\bar\Pi$ corresponds to a cohomological degree $n+1$ function on $T^*[n]L$, which determines a cohomological degree $1$ Hamiltonian vector field $Q_{\bar\Pi}\in \Gamma(T^*[n]L,\T_{T^*[n]L}[1])$. We have:

\begin{prop} Let $L$ be a formal moduli problem and $\bar\Pi\in\Gamma(L, \widehat \Sym^{\bullet} ( \T_{L}[-n]))[n+1]$ an $(n-1)$-shifted polyvector field of mixed degree $j\geq 2$ and of cohomological degree $n+1$. Then
$$ Q_{T^*[n]L}+Q_{\bar\Pi} \quad \in \  \Gamma(T^*[n]L,\T_{T^*[n]L}[1]) $$
is square zero if and only if $\bar\Pi$ defines an $(n-1)$-shifted homotopy Poisson structure on $L$, where $Q_{T^*[n]L}$ is the cohomological vector field defining the $L_\infty$-algebra structure on $\T_o[-1](T^*[n]L)$.
\end{prop}
\begin{proof}
This follows immediately from the preceding proof and discussion.
\end{proof}

In the case where $\bar\Pi$ does define an $(n-1)$-shifted homotopy Poisson structure, the above cohomological vector field defines a new $L_\infty$ algebra yielding a formal moduli problem deforming $T^*[n]L$, which is a twisted cotangent bundle, denoted by $T^*_{\bar\Pi}[n]L $. Further, since the deformation is Hamiltonian, the deformed cohomological vector field is again symplectic for the $n$-shifted symplectic structure on $T^*[n]L$ and thus the twisted cotangent bundle $T^*_{\bar\Pi}[n]L$ is also an $n$-shifted symplectic formal moduli problem. Also, note that such twistings preserve the zero section map $\sigma_0:L\to T^*_{\bar\Pi}[n]L$.

In terms of this space, the model for the $(n-1)$-shifted homotopy Poisson structure on a regular embedded Lagrangian $L\into X$, for $X$ an $n$-shifted symplectic formal moduli problem, can be understood as follows: there is a unique twisted cotangent bundle $T^*_{\bar\Pi}L$ deforming $T^*L$ which has total space isomorphic to $X$ as an $n$-shifted symplectic formal moduli problem, intertwining the maps $L\into X$ and $\sigma_0:L\to T_{\bar\Pi}^*[n]L$; the Hamiltonian function corresponding to the deformation of the cohomological vector field defining this twisting is precisely the same data as an $(n-1)$-shifted homotopy Poisson structure on $L$. In this sense, the construction of $T^*_{\bar\Pi}[n]L$ from the Poisson structure on $L$ is an inverse to the procedure of constructing the $(n-1)$-shifted homotopy Poisson structure on $L$ described above for regular embedded Lagrangians $L\into X$.

\begin{prop} Let $X$ be an $n$-shifted symplectic formal moduli problem and $f:L\to X$ a Lagrangian for $X$ with $(n-1)$-shifted Poisson structure $\bar\Pi$. Then there exists a map $\tilde f:T^*_{\bar\Pi}[n]L\to X$ of $n$-shifted symplectic formal moduli problems such that
\[\xymatrix{
L \ar[r]^-{\sigma_0} \ar[dr]_f & T^*_{\bar\Pi}[n]L \ar[d]^{{\tilde f}}  \\
  & X
}\]
commutes.
\end{prop}

Motivated by this universal property, we make the following definition:

\begin{defn} Let $L$ be a formal moduli problem with $(n-1)$-shifted homotopy Poisson structure $\bar\Pi$. The \emph{higher Poisson centre} $Z_{\bar\Pi} (L)$ of $L$ is the $n$-shifted symplectic formal moduli problem
$$Z_{\bar\Pi}(L)= T^*_{\bar\Pi}[n]L.$$
\end{defn}

\begin{eg}\label{ex:Poisson centre of symplectic}
Let $X$ be a formal moduli problem with a strictly local $(n-1)$-shifted symplectic structure $\omega\in \Gamma(X,\wedge^2 \LL_X[n-1])$, which gives rise to $\Pi_\omega \colon \Gamma(X, \mr{Sym}^2(\T_X[-n]) )[n+1] $. Following the construction above, this leads to $Q_{\Pi_\omega} \in \Gamma( T^*[n] X, \T_{T^*[n]X} )[1] $ understood as the differential. More concretely, $\omega$ is given by a symmetric linear map $\omega \colon \g\to \g^*[n-2]$ and in particular $\Pi_\omega$ by $\Pi_\omega \colon \g^*[n-2]\to \g$ which is an isomorphism as a vector space. Under the identification $\T_o (T^*[n]X) = \g \ltimes \g^*[n-2]$, the differential $Q_{\Pi_\omega}$ is given by an isomorphism $\g^*[n-2]\to  \g$. In particular, it is contractible and hence has homotopically trivial $L_\infty$ algebra structure. This observation is used for proving Proposition \ref{prop:universal bulk of non-degenerate}.
\end{eg}

\begin{rmk}\label{rmk:global interpretation}
In terms of the global algebraic language, as $X$ is $(n-1)$-shifted symplectic, one has $T^*[n]X \simeq T[1] X$ which is the Dolbeault stack $X_{\mr{Dol}} $. Then adding a non-degenerate Poisson bivector amounts to turning on the de Rham differential on the Dolbeault stack $T[1]X$ to obtain the de Rham stack $X_{\mr{DR}}$, which has the trivial tangent complex. This means that the corresponding theory is trivial in a perturbative description. However, this is not necessarily the case at a \emph{nonperturbative} level. For instance, the Kapustin--Witten A-twist as computed by \cite{EY} involves the de Rham stack and hence is perturbatively trivial but encodes interesting information to be S-dual to the B-twist which is nontrivial even at the perturbative level. Indeed, from the nature of S-duality, one shouldn't expect perturbative information of two dual theories to be comparable. 
\end{rmk}

\subsection{The Phase Space of a Classical Field Theory on a Manifold with Boundary}\label{PhS}

In this subsection, we define the phase space of a classical field theory on a manifold with boundary, in terms of the formalism of local moduli problems, using a simplifying assumption that the classical field theory is topological in the direction normal to the boundary of the manifold. This definition is essentially equivalent to the construction of the boundary BFV theory in \cite{CattaneoMnevReshetikhin}.

Let $M$ be a manifold with boundary, $N=\del M$, $M^\circ=M\setminus N$, and $U_\e\cong N\times (0,\e) $ the interior of a collar neighbourhood of $\del M$ in $M$. We define a classical field theory $(\X,\omega)$ on $M$ as simply one defined on $M^\circ$, and will define the phase space of $(\X,\omega)$ on $N$ as a local moduli problem $\X^\del$ equipped with an appropriately local $0$-shifted symplectic structure. To simplify this problem, we restrict our attention to field theories which are topological in the direction normal to $\Sigma$, in a precise sense described below, which will in particular ensure the $0$-shifted symplectic structure on $X^\del$ is strictly local. We first state a condition on the underlying local moduli problem of the classical field theory, followed by an additional constraint on the symplectic form in the case that the preceding condition is satisfied:

\begin{defn} Let $\X$ a local moduli problem on a manifold $M$ with boundary $\del M=N$. Then $\X$ is \emph{topological normal to $N$} if
$$\X|_{U_\e} \cong \widehat{\underline{\mr{Map}}}_o((0,\e)_\dR,\X^\del)$$
as local moduli problems on $U_\e\times N$, for a local moduli problem $\X^\del$ on $N$.
\end{defn}

Let $L$ denote the vector bundle on $M$ underlying $\scg=\T_o[-1]\X$, and $L^\del$ the vector bundle on $N$ underlying $\scg^\del=\T_o[-1]\X^\del$. Then, unpacking this definition, we have that
$$L|_{U_\e}\cong\wedge^\bullet T^*_{(0,\e)}\boxtimes L^\del $$
and moreover
$$
l_1|_{U_\e} = l_1^\del\otimes \id_{\Omega^\bullet_{(0,\e)}} + \id_{\scg^\del}\otimes d_{(0,\e)} \quad\quad\text{and}\quad\quad  l_n|_{U_\e}  = l_n^\del\otimes \mu^n \quad \text{for }n\geq 2,
$$
where $\{l_n^\del:(\scg^\del)^{\otimes n}\to \scg^\del[2-n]\}_{n\in\N}$ is the family of polydifferential operators over $N$ defining the local $L_\infty$ structure on $\scg^\del$, and $\mu^n:(\Omega^\bullet_{(0,\e)})^{\otimes n}\to\Omega^\bullet_{(0,\e)}$ denotes the multiplication of $n$ elements.

Further, we make the following definition for a classical field theory:

\begin{defn}
Let $(\X,\omega)$ a classical field theory on $M$ such that $\X$ is topological normal to $N$. Then $(\X,\omega)$ is \emph{topological normal to $N$} if
$$\omega|_{U_\e}= I_{(0,\e)}\boxtimes\omega^\del :\wedge^\bullet T^*_{(0,\e)} \boxtimes  L^\del  \to (\wedge^\bullet T^*_{(0,\e)})^! \boxtimes (L^\del)^! [-3], $$
for $\omega^\del:L^\del\to (L^\del)^![-2]$ a strictly local $0$-shifted symplectic structure on the local moduli problem $\X^\del$, and where $I_{(0,\e)}: \wedge^\bullet T^*_{(0,\e)}\to  (\wedge^\bullet T^*_{(0,\e)})^![-1]$ is the isomorphism corresponding to integration.
\end{defn}
Under these conditions, we define the phase space of the classical field theory:

\begin{defn}
Let $(\X,\omega)$ a classical field theory on a manifold $M$ with boundary $\del M=N$, which is topological normal to $N$. The \emph{phase space of $(\X,\omega)$ on $N$} is the local moduli problem $\X^\del$ over $N$ corresponding to $\scg^\del$ above, equipped with the required strictly local $0$-shifted symplectic structure $\omega^\del$ on $\X^\del$.
\end{defn}

Note that the $0$-shifted symplectic variant of Proposition $\ref{slla}$, as discussed in Remark $\ref{sslla}$, gives an equivalence between the data of the $L_\infty$ brackets $\{l_n^\del\}_{n\in\N}$ on $\scg^\del$ and a cohomological degree $1$ local action functional $S^\del\in\calO_\loc(\X^\del)[1]$ satisfying the $0$-shifted classical master equation $\{S^\del,S^\del\}_{\omega^\del}=0$.

\begin{eg}\label{ubtp} For any local moduli problem $\Y$ over $N$ equipped with a strictly local $0$-shifted symplectic structure $\eta$, there exists a classical field theory on $\R_{\geq 0}\times N$ with phase space given by $\Y$, with underlying local moduli problem $\widehat{\underline{\mr{Map}}}_o((\R_{\geq 0})_\dR,\Y)$ and $(-1)$-shifted symplectic form $\omega=I_{\R_{\geq 0}}\boxtimes \eta$. This fact will be used essentially in the construction of the universal bulk theory.
\end{eg}

\begin{eg} From the mapping stack adjunction, an AKSZ type theory on a manifold $M$ with boundary $\del M=N$, defined by $\X=\widehat{\underline{\mr{Map}}}_o(M_\dR,B\g)$ and $\omega=I_M\boxtimes \eta$, has phase space $\X^\del=\widehat{\underline{\mr{Map}}}_o(N_\dR,B\g)$ with $\omega^\del=I_N\boxtimes \eta$.
\end{eg}

\subsection{Local Boundary Conditions for Classical Field Theories}

In this subsection, we define the notion of a local boundary condition for a classical field theory on a manifold with boundary. Let $M$ be a manifold with boundary $N=\del M$, $(X,\omega)$ a classical field theory on $M$ which is topological normal to $N$, and $(\X^\del,\omega^\del)$ the phase space of $(\X,\omega)$ on $N$.

\begin{defn} A \emph{regular embedded local boundary condition} for $(X,\omega)$ on $N$ is:
\begin{itemize}
\item a local moduli problem $\scL_+$ over $N$, with $\scl_+=\T_o[-1]\scL_+$ the defining local $L_\infty$ algebra with underlying vector bundle denoted $L_+$, and

\item a homotopically strict, strictly local map $F:\scL_+\to\X^\del$ of local moduli problems over $N$, with underlying bundle map $f:L_+\to L^\del $ injective,
\end{itemize}
such that there exists a choice of complementary subbundle $L_-\into L^\del$ to $L_+$ in $L^\del$, making $L_+$ and $L_-$ Lagrangian subbundles of the vector bundle $L^\del$ with respect to the strictly local $0$-shifted symplectic structure $\omega^\del:L^\del\otimes L^\del\to \Dens_N$, in the sense that $\omega^\del|_{L_\pm\otimes L_\pm}=0$.
\end{defn}

Unpacking this definition, we have that the $L_\infty$ brackets $\{l^+_n:\scl_+^{\otimes n}\to\scl_+[2-n]\}_{n\in\N}$ defining $\scL_+$ are given by
$$l_n^+=l_n^\del\circ f^{\otimes n} $$
where $\{l_n^\del\}_{n\in\N}$ are the $L_\infty$ brackets defining $\X^\del$. Further, since $\omega^\del:L^\del\to (L^\del)^![-2]$ is an isomorphism, the Lagrangian condition implies that
$$ \omega^\del|_{L_+}:L_+\xrightarrow{\cong} L_-^![-2] \quad\quad\text{and}\quad\quad \omega^\del|_{L_-}:L_-\xrightarrow{\cong} L_+^![-2]$$
have the stated codomains and moreover are isomorphisms of vector bundles on $N$.

\begin{eg} Recall that an AKSZ type theory defined by $\X=\widehat{\underline{\mr{Map}}}_o(M_\dR,B\g)$ and $\omega=I_M\boxtimes \eta$ has phase space $\X^\del=\widehat{\underline{\mr{Map}}}_o(N_\dR,B\g)$ with $\omega^\del=I_N\boxtimes \eta$. Given a regular embedded Lagrangian $B\el_+\into B\g$ and a complementary subspace $\el_-$ to $\el_+\into \g$, the local moduli problem $\scL_+=\widehat{\underline{\mr{Map}}}_o(N_\dR,B\el_+)$ has a natural homotopically strict, strictly local map to $\X^\del$ with injective underlying bundle map $L_+=\wedge^\bullet T^*_N\otimes \el_+\into L=\wedge^\bullet T^*_N\otimes \g$ and a complementary Lagrangian subbundle $L_-= \wedge^\bullet T^*_N\otimes \el_-$, defining a regular embedded local boundary condition. 
\end{eg}

\subsection{Construction of the Boundary Theory}

In this subsection we explain our main construction, which proceeds precisely as in the finite dimensional toy model given in Subsection $\ref{TM}$:

\begin{theo}\label{BTC}

Let $(\X,\omega)$ be a classical field theory on a manifold $M$ with boundary $\del M=N$ such that $(\X,\omega)$ is topological normal to $N$, and $\scL_+\into \X^\del$ a regular embedded local boundary condition. Then there is a natural local $(-1)$-shifted homotopy Poisson structure $\bar\Pi$ on $\scL_+$.

Further, there is an isomorphism of sheaves of $C^\infty_N$ modules, and of sheaves of Lie algebras,
$$\calO_\loc(\X^\del)\cong \Gamma_{m\loc}(\scL_+, \widehat\Sym^\bullet \T_{\scL_+})$$
under which one has
$$S^\del_1 \cong Q_{\scl_+}\quad \in\  \Gamma_\loc(\scL_+, \T_{\scL_+}[1]) \qquad \text{and}\qquad  \sum_{j=2}^\infty S^\del_j\cong \bar\Pi \quad \in \ \Gamma_{m\loc}(\scL_+, \widehat \Sym^\bullet \T_{\scL_+})[1]$$
where
$$ S^\del=\sum_j S^\del_j \qquad\text{for}\qquad  S^\del_j \in \Gamma_{m\loc}(\scL_+, \widehat\Sym^\bullet \T_{\scL_+}) $$
is the decomposition in polyvector field degree of the image under the preceding isomorphism of the cohomological degree $1$ action functional $S^\del\in \calO_\loc(\X^\del)[1]$ describing the local $L_\infty$ structure on $\scg^\del$, and $Q_{\scl_+}$ is the cohomological vector field defining the local $L_\infty$ structure on $\scl_+=\T_o[-1]\scL_+$.
\end{theo}

If this Poisson structure $\bar\Pi$ is strict, given by $\Pi\in \Gamma_{m\loc}(\scL_+, \Sym^2\T_{\scL_+})$, then $(\scL_+,\Pi)$ defines a degenerate classical field theory on $N$, which we call the $\emph{boundary theory}$ for $(\X,\omega)$ corresponding to the boundary condition $\scL_+\into \X^\del$. 

\begin{proof}

Let $\scg=\T_o[-1]\X,\scg^\del=\T_o[-1]\X^\del$ and $\scl_+=\T_o[-1]\scL_+$, and let $\scl_-$ denote the sheaf of sections of a fixed complementary Lagrangian subbundle $L_-$ to $L_+$ in $L^\del$. Then we have a direct sum decomposition as sheaves of $C^\infty_N$-modules $\scg^\del= \scl_+\oplus \scl_-$, and from the bundle isomorphism $\omega|_{L_-}:L_-\to L_+^![-2]$ we further obtain $\scg^\del \cong \scl_+ \oplus \scl_+^![-2]$ in the same category. This gives a further isomorphism of sheaves of $C^\infty_N$-modules
$$ \calO_\loc(\X^\del) \cong  \Dens_N\otimes_{\calD_N} \widehat \Sym^\bullet(\calJ(\scl_+)^\vee[-1])\otimes_{C^\infty_N}\widehat \Sym^\bullet( \calJ(\scl_+^!)^\vee[1])$$
Now, applying to one of the symmetric factors the natural isomorphism $\calJ(\scl_+^!)^\vee \cong \Diff(\scl_+^!,C^\infty_N)$ and taking the formal adjoint yields $\Diff(\Dens_N, \scl_+)$. Applying the isomorphism $\Dens_N\otimes_{\calD_N} \Diff(\Dens_N, \scl_+)\cong \scl_+$, and taking a symmetrizing sum over symmetric factors yields an isomorphism
$$ \Dens_M\otimes_{\calD_M} \widehat \Sym^k(\calJ(\scl_+)^\vee[-1])\otimes_{C^\infty_M}\widehat \Sym^j( \calJ(\scl_+^!)^\vee[1]) \cong \Poly(\scl_+[1]^{\otimes k}\otimes \scl_+^![-1]^{\otimes (j-1)},\scl_+[1])_{S_k\times S_j} $$
and thus an isomorphism of sheaves of $C^\infty_N$ modules $\calO_\loc(\X^\del)\cong \Gamma_{m\loc}(\scL_+, \widehat\Sym^\bullet \T_{\scL_+})$ as claimed. 

Next, we show this is also an isomorphism of sheaves of Lie algebras with respect to the ($0$-shifted) Poisson bracket on $\calO_\loc(\X^\del)$ determined by the strictly local $0$-shifted symplectic structure on $\X^\del$ and the Schouten bracket on $\Gamma_{m\loc}(\scL_+, \widehat\Sym^\bullet \T_{\scL_+})$. The former is determined by its value on the space of local functionals $\calO_\loc^1(\X^\del)$ which are linear in $\scL_-$, while the later is determined correspondingly by its value on the local vector fields $\Gamma_{m\loc}(\scL_+,\T_{\scL_+})$. Moreover, the isomorphism constructed above identifying $\calO_{\loc}^1(\X^\del)\cong \Gamma_\loc(\scL_+,\T_{\scL_+})$ agrees with the map
$$\Pi_{\omega^\del} \circ d_\dR|_{\calO_{\loc}^1(\X^\del)}: \calO_\loc^1(\X^\del) \to \Gamma_\loc(\scL_+,\T_{\scL_+})\subset \Gamma_\loc(\X^\del,\T_{\X^\del})\ ,$$
as taking the formal adjoint of the symmetric factors and contracting over $\calD_M$, as above, agrees with the definition of $d_\dR$ given in the proof of Proposition $\ref{ldrd}$ modulo the identification given by $\Pi_\omega$. Now, the Poisson bracket on $\calO_\loc(\X^\del)$ is defined by
$$\{f,g\}= \langle \Pi_{\omega^\del}(d_\dR f), d_\dR g\rangle_{\T_{\X^\del}}= \langle \Pi_{\omega^\del}(d_\dR f), \omega\circ \Pi_\omega( d_\dR g) \rangle_{\T_{\X^\del}},$$
which agrees with the contraction of differential operators defining the Schouten bracket of $\Pi_\omega( d_\dR f)$ and $\Pi_\omega( d_\dR g)$, as desired.

Having established the above isomorphism, the remainder of the proof proceeds precisely as in the finite dimensional case given in Subsection $\ref{TM}$.
\end{proof}

\subsection{Local Higher Poisson Centres and Universal Bulk Theories}

In this subsection we define the notion of the local higher Poisson centre of a local moduli problem equipped with a local $(-1)$-shifted homotopy Poisson structure. Further, we formulate the definition of the universal bulk theory corresponding to such an object.

Let $\scL$ be a local moduli problem over $N$ and $\bar\Pi \in \Gamma_{m\loc}(\scL, \widehat\Sym^\bullet\T_\scL)[1]$ a local $(-1)$-shifted homotopy Poisson structure on $\scL$. Further, let $T^*\scL:=|\LL_\scL|$ denote the local moduli problem given by the total space of the cotangent bundle to $\scL$.

Note that we have an isomorphism of precosheaves on $N$ of dg Lie algebras
$$\calO_\loc(T^*\scL) \cong \Gamma_{m\loc}(\scL, \widehat\Sym^\bullet\T_\scL)$$
where the Lie algebra structure on $\calO_\loc(T^*\scL)$ is inherited from the local $0$-shifted Poisson structure corresponding to the canonical strictly local $0$-shifted symplectic structure on $T^*\scL$, and the Lie algebra structure on $\Gamma_{m\loc}(\scL, \widehat\Sym^\bullet\T_\scL)$ is given by the Schouten bracket. We can thus interpret $\bar\Pi$ as a cohomological degree $1$ local action functional, with corresponding Hamiltonian vector field $Q_{\bar\Pi}\in \Gamma_\loc( T^*\scL,\T_{T^*\scL}[1])$ and we have:
\begin{prop}Let $\scL$ be local moduli problem over $N$ and $\bar\Pi \in \Gamma_{m\loc}(\scL, \widehat\Sym^\bullet\T_\scL)[1]$. Then
$$ Q_{ T^*\scL} + Q_{\bar\Pi} \quad \in \ \Gamma_\loc( T^*\scL,\T_{T^*\scL}[1])$$
is square zero if and only if $\bar\Pi$ defines a local $(n-1)$-shifted homotopy Poisson structure on $\scL$, where $Q_{ T^*\scL}$ is the local vector field defining the local $L_\infty$ structure on $\T_o[-1] T^*\scL$.
\end{prop}

In the case where $\bar\Pi$ does define a $(-1)$-shifted homotopy Poisson structure on $\scL$, the above local cohomological vector field defines a local $L_\infty$ algebra on $N$, and we let $T^*_{\bar\Pi}\scL$ denote the corresponding local moduli problem on $N$. As in the finite dimensional toy model, the canonical strictly local, $0$-shifted symplectic structure $\omega$ on $T^*\scL$, introduced in Example $\ref{csf}$, induces another $0$-shifted symplectic structure on $T^*_{\bar\Pi}\scL$. Further, we have a homotopically strict, strictly local map $\sigma_0:\scL\to T^*_{\bar\Pi}\scL$ of local moduli problems on $N$.

\begin{defn}Let $\scL$ be local moduli problem over $N$ and $\bar\Pi \in \Gamma_{m\loc}(\scL, \widehat\Sym^\bullet\T_\scL)[1]$ a local $(-1)$-shifted homotopy Poisson structure on $\scL$. The \emph{local higher Poisson centre} is the local moduli problem
$$\calZ_{\bar\Pi}(\scL)=T^*_{\bar\Pi}\scL$$
over $N$, equipped with its induced strictly local $0$-shifted symplectic structure.
\end{defn}

As discussed in Example $\ref{ubtp}$, for any local moduli problem $\Y$ on a manifold $N$, together with a strictly local, $0$-shifted symplectic form $\omega$ on $\Y$, there is a natural classical field theory on $N\times \R_{\geq 0}$ with underlying moduli problem $\underline{\mr{Map}}(\R_{\geq 0},\Y)$, yielding $(\Y,\omega)$ as its phase space. We now formulate the definition of the universal bulk theory:

\begin{defn} Let $\scL$ be local moduli problem over $N$ and $\bar\Pi \in \Gamma_{m\loc}(\scL, \widehat\Sym^\bullet \T_\scL)[1]$ a local $(n-1)$-shifted homotopy Poisson structure on $\scL$. We define the \emph{universal bulk theory} corresponding to $(\scL,\bar\Pi)$ to be
$$ \calU_{\bar\Pi}(\scL) : = \widehat{\underline{\mr{Map}}}_o (\R_{\geq 0},\calZ_{\bar\Pi}(\scL))$$
equipped with the $(-1)$-shifted symplectic structure $\omega_\calU=I_{\R_{\geq 0}}\boxtimes \tilde \omega$, where $\calZ_{\bar\Pi}(\scL)$ is the local higher Poisson centre of $(\scL,\Pi)$ and $\tilde \omega$ is the $0$-shifted symplectic form on $\calZ_{\bar\Pi}(\scL)$ described above.

Further, we define the \emph{canonical boundary condition} for $\calU_{\bar\Pi}(\scL)$ as the $0$-section map $\sigma_0:\scL\to \calZ_{\bar\Pi}(\scL)$ described above.
\end{defn}

\begin{prop}\label{prop:universal bulk of non-degenerate}
Let $\X$ be a non-degenerate classical field theory. The universal bulk theory of $\X$ has trivial $L_\infty$ algebra structure up to homotopy.
\end{prop}

\begin{proof}
See Example \ref{ex:Poisson centre of symplectic} \textit{Mutatis mutandis}.
\end{proof}

\begin{rmk}
While there could be several different boundary theories one could consider for a given bulk theory, there is a unique universal bulk theory for a given degenerate classical field theory. It is noteworthy that the universal bulk theory of the universal bulk theory is always trivial by the proposition. One might read this as a version of the fundamental equation $d^2=0$ of homological algebra in the context of perturbative field theory. Note, though, that this is not the case at a nonperturbative level (see Remark \ref{rmk:global interpretation}).
\end{rmk}

\newpage

\section{Examples of Classical Boundary Theories}\label{section:examples}

In this section, we present several examples of the formalism developed in the previous sections. In each example, we proceed by first motivating the spaces involved in terms of global derived algebraic geometry, then explaining the precise description of the field theory in terms of the formal local language introduced in the preceding sections.

\subsection{Topological Classical Mechanics}

Our first example of a classical field theory is topological classical mechanics, the classical limit of topological quantum mechanics. This example demonstrates the main ideas of the previous sections in the simplest possible setting.

In the simplest case, for $V$ a ($0$-shifted) symplectic vector space, topological classical mechanics valued in $V$ is a 1-dimensional AKSZ-type classical field theory described by $\underline{\mr{Map}}( M^1_{\mr{dR}}, V  )$. The phase space of this theory on $\del \R_{\geq 0 } = \{0\}$ is just the symplectic vector space $V$ and its boundary conditions are Lagrangian subspaces $L\into V$. More generally, this construction can be applied globally for any target symplectic variety $X$, and working perturbatively around the constant map to a point $x\in X$ is equivalent to the linear case with $V=T_x X$.

For the local, formal description of the classical field theory, let $\g=V[-1]$ be the shift of $V$ by $-1$, viewed as a trivial $L_\infty$ algebra; this should be thought of as the $(-1)$-shifted tangent complex $\T_x[-1]X$ to a point $x\in X$. Then $\scg=\Omega^\bullet_{M^1} \otimes \g$ defines a local $L_\infty$ algebra and the symplectic form on $V$ gives rise to a strictly local $(-1)$-shifted symplectic structure on $\scg$ via the usual AKSZ construction.

More generally, let $\g$ be an arbitrary $L_\infty$ algebra and $\omega:\g\to \g^*[-2]$ define a $0$-shifted symplectic structure on the formal moduli problem $X$ corresponding to $\g$. Then $\underline{\widehat{\mr{Map}}}_o( M^1_{\mr{dR}}, X )$ has a strictly local $(-1)$-shifted symplectic structure coming from the AKSZ construction, and the resulting classical field theory describes topological classical mechanics valued in a derived stack $\tilde X$, in perturbation theory around the constant map to a geometric point $x\in \tilde X$, or equivalently, topological classical mechanics valued in the formal neighbourhood of $x$ in $\tilde X$; this theory was studied in detail in this language in \cite{GradyGwilliam} \cite{GradyLiLi}.

The phase space of this theory on $\R_{\geq 0}$ is the $0$-shifted symplectic formal moduli problem $X$ corresponding to the $L_\infty$ algebra $\g$, and its regular embedded boundary conditions are just strict derived Lagrangians $L_+\into X$. Further, the boundary theory corresponding to a given boundary condition $L\to X$ is just the formal moduli problem $L$, equipped with its $(-1)$-shifted Poisson structure. Note that since $\del\R_{\geq 0}=\{0\}$ is $0$-dimensional, all of the locality conditions on the phase space and boundary conditions are trivial, and thus the boundary theory construction in this case reduces to the finite dimensional model presented in Subsection $\ref{TM}$.

\subsubsection{Topological Poisson Mechanics and the Poisson $\sigma$-Model}

We can generalize the above construction to consider topological classical mechanics valued in a formal moduli problem $X$ with an arbitrary $0$-shifted Poisson structure $\Pi\in \Gamma(X,\widehat\Sym^2( \T_L[-1]))[2]$, yielding a degenerate classical field theory of (Poisson) AKSZ type, with underlying local moduli problem $\underline{\widehat{\mr{Map}}}_o( M^1_{\mr{dR}}, X )$ and (strictly) local $(-1)$-shifted Poisson structure given by
$$ I_M^{-1}\boxtimes \Pi: (\wedge^\bullet T^*_M)^!\boxtimes \g^*[-1] \to \wedge^\bullet T^*_M\boxtimes \g[2] $$

The universal bulk theory constructed from topological Poisson mechanics is given by $\underline{\mr{Map}}( M^1_{\mr{dR}} \times \R_{\geq 0} , T_\pi^*[1] X )$, which is the well-studied Poisson $\sigma$-model \cite{Kontsevich} \cite{CattaneoFelder} \cite{CattaneoFelderII}.

\begin{rmk}
There are a few theories which look similar to the Poisson $\sigma$-model, and also admit descriptions in this formalism:
\begin{itemize}
\item (2-dimensional topological Yang--Mills theory) In the case $\g= \mr{Lie}(G)$ for a semi-simple group $G$, the classical field theory described by $\underline{\mr{Map}}( M^2_{\mr{dR}}, T^*[1] B\g  )$ is called 2-dimensional topological Yang--Mills theory. This theory is going to play a crucial role for our description of derived Hamiltonian reduction from the physical point of view. Of course, nothing prevents one from considering an $L_\infty$ algebra $\g$ here.
\item (B-model) For a Calabi--Yau manifold $X$, the classical field theory underlying the B-model is described by $\underline{\mr{Map}}( M^2_{\mr{dR}}, T^*[1] X  )$. \
\end{itemize}
\end{rmk}

\subsection{Chern--Simons Theory}

Chern--Simons theory with gauge group $G$, a semi-simple complex group, is a three dimensional AKSZ type field theory,  with space of solutions to the Euler--Lagrange equations on a 3-manifold $M$ given by $\mr{Loc}_G(M)$. The phase space of Chern--Simons on a manifold $M$ with boundary $\del M= \Sigma$ is $\mr{Loc}_G(\Sigma)$. Recall from Remark \ref{rmk:algebraic vs smooth} that in the perturbative setting, $\mr{Loc}_G(\Sigma)$ is indistinguishable from $\mr{Flat}_G(\Sigma)$, so we will work with $\mr{Flat}_G(\Sigma)$ as the phase space on $\Sigma$. In this way, the boundary conditions we consider will admit natural descriptions and their relationship with the Kapustin--Witten twists of $\calN=4$ super Yang-Mills and the geometric Langlands program will become clear. Indeed, we claim that Kapustin--Witten theory is the universal bulk theory of Chern--Simons theory; however, this discussion is deferred to a seperate subsection on Kapustin--Witten theory, which also admits other interesting boundary theories.

As the moduli space of solutions to the equations of motion on $U\subset M$ is $\mr{Loc}_G(U)$, the corresponding presheaf $\X$ of formal moduli problems on $M$ is given by $U \mapsto \X_U =  \mr{Loc}_G(U)^\wedge_{P|_U}$ for a fixed flat bundle $P$ on $M$, which by abuse of notation is regarded as a point in $\mr{Loc}_G(U)$ for each open set $U$; for example, the trivial flat bundle defines a pointing of $\mr{Loc}_G(U)$ for each $U\subset M$. Thus, Chern--Simons theory is described in the formal local language of the previous sections by the local moduli problem $\widehat{\underline{\mr{Map}}}_o(M_\dR,B\g)$ with corresponding local $L_\infty$ algebra
\[ \scg_{\text{CS}}= \Omega_M^\bullet \otimes \g \qquad \text{with} \quad l_1 = d_M \otimes \id_\g, \; l_2 = \wedge \otimes [\;,\;], \text{ and } l_n=0 \text{ for } n\geq 3.\]

The AKSZ type symplectic pairing $\omega \colon L\to L^![-3]$ or $\omega \colon \Omega^\bullet_M \otimes \g \to ( \Omega^\bullet_M \otimes \g )^\vee \otimes \mr{Dens}_M[-3]$ is given by $\alpha \mapsto (\beta \mapsto \alpha \wedge \beta)$ using the integration pairing on $\Omega^\bullet_M$ and a symmetric invariant non-degenerate pairing on $\g$. The action functional defining the above local $L_\infty$ structure via this symplectic pairing is the classical Chern--Simons action
\[S(A)=\int_M \mr{Tr} \left( \frac{1}{2} A\wedge dA+ \frac{1}{6} A\wedge[A, A  ] \right) \] for $A\in\Omega^\bullet(M)\otimes\g$.

The phase space of Chern--Simons theory on a manifold $M$ with boundary $\del M=\Sigma$ is given in the formal local language by the $0$-shifted symplectic local moduli problem $\widehat{\underline{\mr{Map}}}_o(\Sigma_\dR,B\g)$ with corresponding local $L_\infty$ algebra \[\scg^\Sigma = \Omega_\Sigma^\bullet \otimes \g \qquad \text{with} \quad l_1 = d_\Sigma \otimes \id_\g, \; l_2= \wedge \otimes [\;,\;],  \text{ and }l_n=0\text{ for }n\geq 3.\]

\begin{rmk}[Critical Chern--Simons theory]
We also introduce here critical Chern--Simons theory, a 3d classical field theory defined on three manifolds of the form $M^1\times \Sigma$ for a fixed Riemann surface $\Sigma$, which we claim to be the limit of Chern--Simons theory compatible with the complex structure on $\Sigma$. 

In order to understand the dependence on the complex structure clearly, as discussed in Remark \ref{rmk:algebraic vs smooth}, we treat Chern--Simons theory as topological classical mechanics with target $\mr{Flat}_G(\Sigma)$, using the adjunction of a mapping stack:
\[\underline{\mr{Map}}( M^1_{\mr{dR}} \times \Sigma_{\mr{DR}} , BG   ) =\underline{\mr{Map}} (M^1_{\mr{dR}} , \underline{\mr{Map}} (\Sigma_{\mr{DR}}, BG ) )= \underline{\mr{Map}} (M^1_{\mr{dR}} , \mr{Flat}_G(\Sigma) ).\]
From this perspective, we think of the complex level $k\in \mathbb{C}$ of Chern--Simons theory as corresponding to the twisting parameter of $\mr{Flat}_G(\Sigma) = T^*_{\mr{tw}} \mr{Bun}_G(\Sigma)$ and the critical level corresponding to no twisting; thus, we define critical Chern--Simons theory as topological classical mechanics with target $\mr{Higgs}_G(\Sigma) = T^* \mr{Bun}_G(\Sigma)$.

The formal local description of critical Chern--Simons theory is as the corresponding formal mapping space, defined by the local $L_\infty$ algebra
\[\scg_{ \text{cCS} } = \Omega_{M^1}^\bullet \otimes \Omega_{\Sigma}^{\bullet,\bullet} \otimes \g\qquad \text{with} \quad l_1 = d_{M^1} \otimes \id_{\Omega^{\bullet,\bullet}_\Sigma} \otimes \id_\g+ \id_{\Omega^\bullet_{M^1}}\otimes \delbar_\Sigma \otimes \id_\g \ , \; l_2 = \wedge \otimes \wedge \otimes [\;,\;] \]
and $l_n=0$ for $n\geq 3$.

The phase space of critical Chern--Simons on $\R_{\geq 0}\times \Sigma$ is given by $\mr{Flat}_G(\Sigma)$, which has formal local description given by the local $L_\infty$ algebra
\[\scg^\Sigma =  \Omega_{\Sigma}^{\bullet,\bullet} \otimes \g\qquad \text{with} \quad l_1 = \delbar_\Sigma \otimes \id_\g, \;l_2 =  \wedge \otimes [\;,\;],  \text{ and }l_n=0\text{ for }n\geq 3.\] Note that critical Chern--Simons theory is not topological along $\Sigma$.
\end{rmk}

\subsection{Chiral Wess--Zumino--Witten Model}

Recall that the phase space of Chern--Simons theory $\mr{Flat}_G(\Sigma) = T^*_{\mr{tw}} \mr{Bun}_G(\Sigma)$ is a twisted cotangent bundle of the moduli space $\mr{Bun}_G(\Sigma)$ of holomorphic $G$-bundles on $\Sigma$. Its twisted cotangent fibre $ T^*_{\mr{tw} , P} \mr{Bun}_G(\Sigma)$ at given bundle $P\in \mr{Bun}_G(\Sigma)$ is a regular embedded Lagrangian, and hence will define a classical boundary condition for Chern--Simons theory on manifolds $M$ with boundary $\del M =\Sigma$. Presently, we consider the cotangent fibre at the trivial bundle $ T^*_{\mr{tw} , \triv} \mr{Bun}_G(\Sigma)$.

The $(-1)$-shifted tangent complex to $\mr{Bun}_G(\Sigma)$ at the trivial bundle is $(\Omega^{0,\bullet}(\Sigma) \otimes \g , \delbar_\Sigma)$, and that of $ T^*_{\mr{tw} , \mr{triv}} \mr{Bun}_G(\Sigma)$ is given by $(\Omega^{1,\bullet}(\Sigma) \otimes \g[-1], \delbar_\Sigma)$. 

We can decompose the formal local description of the phase space of Chern--Simons in perturbation theory around the trivial bundle, as a sheaf of cochain complexes, as
\[ \scg^\Sigma= (\Omega^\bullet_\Sigma \otimes \g , d) = \xymatrix{
\Omega^{0,0}_\Sigma \otimes \g \ar[r]^-{\delbar_\Sigma}  \ar[dr]_-{\del_\Sigma}  &  \Omega^{0,1}_\Sigma \otimes \g  \ar[dr]^{\del_\Sigma} \\
& \Omega^{1,0}_\Sigma \otimes \g \ar[r]^-{\delbar_\Sigma}    &  \Omega^{1,1}_\Sigma \otimes \g,
}\]
and the local boundary condition corresponding to the Lagrangian $ T^*_{\mr{tw} , \triv} \mr{Bun}_G(\Sigma)$ is given by
\[ \scl = \xymatrix{ ( \Omega^{1,0}_\Sigma \otimes \g \ar[r]^-{\delbar_\Sigma}   &  \Omega^{1,1}_\Sigma \otimes \g).} \]
Thus, the induced boundary theory has underlying local $L_\infty$ algebra
\[ \scg_{\text{WZW}} =  (\Omega^{1,0}_\Sigma\otimes \g)[-1]\oplus(\Omega^{1,1}_\Sigma\otimes\g)[-2] \qquad \text{with} \quad l_1 = \delbar_\Sigma \otimes \id_\g \text{ and }l_n=0\text{ for }n\geq 2\]
Note that this is an abelian local $L_\infty$ algebra, which is the appropriate notion of free field theory in the degenerate context. We let $\scL_\text{WZW}$ denote the corresponding local moduli problem.

\begin{rmk}\label{cdr} On a holomorphic disc $\mathbb{D}_z\into \Sigma$ around any point $z\in \Sigma$, the space of local observables $\calO(\scL_\text{WZW})_{\mathbb{D}_z}$ on $\mathbb{D}_z$ is homotopic to the vector space underlying the classical affine Kac--Moody Poisson vertex algebra, by contracting the Dolbeault resolution.
\end{rmk}

Further, we have:

\begin{prop} 
The local $(-1)$-shifted Poisson structure on the local moduli problem $\scL_\text{WZW}$ over $\Sigma$, as determined by Theorem $\ref{BTC}$, is given by
$$\Pi = \del_\Sigma\otimes \id_\g + \wedge\otimes [\;,\;] $$
where
\begin{align*}
\del_\Sigma\otimes\id_\g : &\  \Omega^{0,\bullet}_\Sigma\otimes \g \to \Omega^{1,\bullet}_\Sigma\otimes \g & \ \in\ \Diff(\scl^!,\scl[2])_{S_2}[1]= \Gamma^0_{m\loc}(\scL,\Sym^2\T_\scL)[1]   \\
\wedge\otimes [\; , \; ]: &\ (\Omega^{1,\bullet}_\Sigma\otimes \g )\otimes  (\Omega^{0,\bullet}_\Sigma \otimes \g) \to \Omega^{1,\bullet}_\Sigma\otimes \g &\  \in\ \Poly( \scl\otimes \scl^!,\scl[1])_{S_2}[1] = \Gamma^1_{m\loc}(\scL,\Sym^2\T_\scL)[1]
\end{align*} under the isomorphism $\scl_- \cong \scl^![-2]$.
\end{prop}
\begin{proof}
Choosing the complementary Lagrangian $\scl_-=\Omega^{0,\bullet}_\Sigma\otimes \g$ to $\scl \into \scg^\Sigma$, we have the decomposition of the Chern--Simons action
\begin{align*}
S^\del(l_-+l_+) & = \int_M \mr{Tr} \left( \frac{1}{2} (l_-+l_+)\wedge (\del_\Sigma+\delbar_\Sigma) (l_-+l_+)+  \frac{1}{6} (l_-+l_+)\wedge[ (l_-+l_+),  (l_-+l_+)  ] \right) \\
& = \int_M \mr{Tr}\left( l_-\wedge \delbar_\Sigma l_+ + \frac{1}{2} l_-\wedge \del_\Sigma l_- +  \frac{1}{2} l_-\wedge [ l_-,l_+]   \right)
\end{align*}
for $ (l_-+l_+) \in(\Omega^{0,\bullet}_\Sigma\oplus \Omega^{1,\bullet}_\Sigma) \otimes\g$, so that
$$S^\del_1(l_-+l_+) =\int_M \mr{Tr}\left( l_-\wedge \delbar_\Sigma l_+\right)  \quad\text{and}\quad S^\del_2(l_-+l_+) = \frac{1}{2} \int_M \mr{Tr}\left( l_-\wedge \del_\Sigma l_- +  l_-\wedge [ l_-,l_+]   \right)$$
and $S^\del_j=0$ for $j\geq 3$.
Under the isomorphism given in Theorem $\ref{BTC}$, the terms in $S^\del_2$ correspond to the two summands of the bivector field $\Pi$ above, as claimed.
\end{proof}

\begin{rmk} This Poisson structure is a $(-1)$-shifted analogue of the standard Poisson structure on the classical affine Kac--Moody Poisson vertex algebra, related via Remark $\ref{cdr}$. As our techniques apply globally over the underlying manifold, the data we recover is the $(-1)$-shifted analogue of a Coisson algebra over $\Sigma$.

Moreover, the underlying local moduli problem corresponding to $\scg_\text{WZW}$ is free, and free BV quantization of this degenerate classical field theory yields the twisted factorization envelope construction of $\cite{CostelloGwilliam}$, for which the associated factorization algebra of quantum observables has been shown to recover the affine Kac--Moody vertex algebra. An explicit proof of this claim will be given in a forthcoming note.
\end{rmk}

\begin{rmk}
If we instead chose the Lagrangian $\scl =\xymatrix{ (\Omega^{0,1}_\Sigma \otimes \g \ar[r]^-{\del_\Sigma}    &  \Omega^{1,1}_\Sigma \otimes \g), }$ this would define the anti-chiral WZW model. By reducing Chern--Simons theory on $\Sigma\times[0,1]$ along the interval, with boundary conditions giving the chiral and anti-chiral WZW models at the ends of the interval, one obtains the full WZW model on $\Sigma$. A proof of this later statement, among other constructions coming from reduced boundary condition configurations, will appear in a forthcoming version of the present paper.
\end{rmk}

\subsection{Chiral Toda Theory}

Another well-studied Lagrangian of the phase space $\mr{Flat}_G(\Sigma)$ of Chern--Simons theory is the moduli space $\mr{Op}_G(\Sigma)$ of opers, which is the space of holomorphic $G$-connections with a fixed $B$-reduction, modulo $N$-valued gauge transformations \cite{BDopers}. An oper is in particular an irreducible, flat $G$-connection and hence the trivial flat connection, which is completely reducible, does not define a point in $\mr{Op}_G(\Sigma)$, so that we cannot detect the local boundary condition induced by the Lagrangian $\mr{Op}_G(\Sigma)$ in perturbation theory about the trivial connection. Instead, for an $\mathfrak{sl}_2$ embedding $(e,f,h)$ in $\g$, we use the element $f$ to define a holomorphic connection $d_f = d + f dz$ on the trivial bundle on $\Sigma$, which defines an oper, and study Chern--Simons theory on $M^1\times \Sigma$ in perturbation theory around the flat connection $\tilde d_f=\id\otimes d_f+d_{M^1}\otimes \id$ on the trivial bundle.

The formal local description of the phase space of Chern--Simons theory perturbed around the connection $\tilde d_f$ is given by the local $L_\infty$ algebra describing the $(-1)$-shifted tangent complex to $\mr{Loc}_G(\Sigma)$ at $d_f$:
$$\scg^\Sigma=\Omega^\bullet_\Sigma \otimes \g\qquad\text{with}\quad l_1=d_f\otimes \id_\g,\ l_2=\wedge\otimes[\;,\;],\ \text{ and }l_n=0\text{ for }n\geq 3$$
where $d_f=d+[f,-]dz$ denotes here the induced connection on the adjoint bundle, for which we use the same notation.

In terms of the chosen $\mathfrak{sl}_2$ triple, the phase space can be decomposed, as a sheaf of cochain complexes, as
\[ \scg^\Sigma = (\Omega^\bullet_\Sigma \otimes \g , d_f) = \xymatrix{(\Omega^{0,\bullet}_\Sigma \otimes \mathfrak{b}_- , \delbar_\Sigma) \ar[rr]^-{\del_\Sigma\otimes\pi_{\n_-} + [f,-]dz }  \ar[drr]^-{\del_\Sigma\otimes\pi_\mathfrak{h}}  &&  (\Omega^{1,\bullet}_\Sigma \otimes \n_-, \delbar_\Sigma)\\
(\Omega^{0,\bullet}_\Sigma \otimes \n , \delbar_\Sigma) \ar[rr]^-{\del_\Sigma\otimes \pi_\n + [f,-]dz}    &&  (\Omega^{1,\bullet}_\Sigma \otimes \mathfrak{b}, \delbar_\Sigma),
}\]
where we have split $\id_\g= (\pi_{\n_-} \oplus \pi_\mathfrak{h}\oplus \pi_\n)$. The $(-1)$-shifted tangent complex at $d_f$ to the moduli of opers gives rise to the local boundary condition
$$\scl =\left( \xymatrix{ (\Omega^{0,\bullet}_\Sigma \otimes \n , \delbar_\Sigma) \ar[rr]^-{\del_\Sigma\otimes \id_\n + [f,-]dz}    &&  (\Omega^{1,\bullet}_\Sigma \otimes \mathfrak{b}, \delbar_\Sigma) }\right)$$
embedded as in the decomposition above. Thus, the resulting local $L_\infty$ algebra for the boundary theory is given by
\[\scg_{\text{Toda}} = (\Omega^{0,\bullet}_\Sigma \otimes \n) \oplus (\Omega^{1,\bullet}_\Sigma \otimes \mathfrak{b} )[-1] \qquad \text{with} \quad l_1 = d_f \otimes \id_\g, \; l_2 = \wedge \otimes [\;,\;],  \text{ and }l_n=0\text{ for }n\geq 3 ,\]
where $[\;,\;]$ denotes the restriction of the bracket on $\g$ to $\n$. We let $\scL_\text{Toda}$ denote the corresponding local moduli problem.

\begin{rmk}\label{cwa}
Evaluating on a holomorphic disc $\mathbb{D}_z$ and contracting the Dolbeault resolutions, as in Remark \ref{cdr}, the local observables $\calO(\scL_\text{Toda})_{\mathbb{D}_z}$ are homotopic to the BRST complex for the classical Drinfeld--Sokolov reduction of the affine Kac--Moody Poisson vertex algebra by the local Lie algebra of holomorphic, $\n$-valued infinitesimal gauge transformations; this complex has cohomology isomorphic to the vector space underlying the classical affine $\calW$ algebra $\calW_\infty(\g)$; see for example $\cite{DSKV} \cite{Valeri}$.
\end{rmk}

We now calculate the $(-1)$-shifted local Poisson structure on $\scL_\text{Toda}$:
\begin{prop}
The local $(-1)$-shifted Poisson structure on the local moduli problem $\scL_\text{Toda}$ over $\Sigma$, as determined by Theorem $\ref{BTC}$, is given by
$$\Pi = \del_\Sigma\otimes \id_\mathfrak{h} +\wedge_{10}\otimes (\pi_\mathfrak{b}\circ [\; , \; ]_{\mathfrak{b} \mathfrak{b}_-}) +2\wedge_{01}\otimes (\pi_\mathfrak{b}\circ  [\; , \; ]_{\n \n_-})$$
where
\begin{align*}
\del_\Sigma\otimes \id_\mathfrak{h}: &\  \Omega^{0,\bullet}_\Sigma\otimes \mathfrak{h} \to \Omega^{1,\bullet}_\Sigma\otimes \mathfrak{h} & \ \in\ \Diff(\scl^!,\scl[2])_{S_2}[1]= \Gamma^0_{m\loc}(\scL,\Sym^2\T_\scL)[1]   \\
\wedge_{10}\otimes (\pi_\mathfrak{b}\circ [\; , \; ]_{\mathfrak{b} \mathfrak{b}_-}): &\ ( \Omega^{1,\bullet}_\Sigma\otimes \mathfrak{b} )\otimes  (\Omega^{0,\bullet}_\Sigma \otimes \mathfrak{b}_-)  \to \Omega^{1,\bullet}_\Sigma\otimes \mathfrak{b} &\  \in\ \Poly( \scl\otimes \scl^!,\scl[1])_{S_2}[1] = \Gamma^1_{m\loc}(\scL,\Sym^2\T_\scL)[1] \\
\wedge_{01}\otimes (\pi_\mathfrak{b}\circ  [\; , \; ]_{\n \n_-}) : &\  (\Omega^{0,\bullet}_\Sigma \otimes \n) \otimes( \Omega^{1,\bullet}_\Sigma\otimes \n_- )\to \Omega^{1,\bullet}_\Sigma\otimes \mathfrak{b} &\  \in\ \Poly( \scl\otimes \scl^!,\scl[1])_{S_2}[1] = \Gamma^1_{m\loc}(\scL,\Sym^2\T_\scL)[1]
\end{align*}
under the isomorphism $\scl_- \cong \scl^![-2]$.
\end{prop}

\begin{proof}
Choosing the complementary Lagrangian subbundle with sheaf of sections $\scl_-= (\Omega^{0,\bullet}_\Sigma \otimes \mathfrak{b}_- \oplus \Omega^{1,\bullet}_\Sigma \otimes \n_-)$ to $\scl\into \scg$, we have the decomposition of the phase space action
\begin{align*}
S^\del(b_-+n_-+n_++b_+ ) & = \frac{1}{2} \int_M \mr{Tr} \left(  (b_-+n_-+n_++b_+)\wedge d_f (b_-+n_-+n_++b_+)\right)  \\
& \qquad\qquad + \frac{1}{6} \int_M \mr{Tr}\left( (b_-+n_-+n_++b_+)\wedge[b_-+n_-+n_++b_+,b_-+n_-+n_++b_+]  \right) \\
& = \int_M \mr{Tr}\left(  b_-\wedge (\del_\Sigma +[f,-]dz) n_+ + n_-\wedge\delbar_\Sigma n_+ + \frac{1}{2} b_-\wedge \del_\Sigma b_-  \right)  \\
&\qquad\qquad + \int_M \mr{Tr}\left( b_-\wedge[n_+,b_+] + \frac{1}{2} n_-\wedge[n_+,n_+] + \frac{1}{2} b_-\wedge[b_-,b_+] + b_-\wedge[n_-,n_+] \right)
\end{align*}
for $(b_-+n_-+n_++b_+) \in (\Omega^{0,\bullet}_\Sigma \otimes \mathfrak{b}_-)\oplus (\Omega^{1,\bullet}_\Sigma \otimes \n_-)\oplus (\Omega^{0,\bullet}_\Sigma \otimes \n)\oplus (\Omega^{1,\bullet}_\Sigma \otimes \mathfrak{b})$, so that
\begin{align*}
S^\del_1(b_-+n_-+n_++b_+) & = \int_M \mr{Tr}\left(  b_-\wedge (  (\del_\Sigma +[f,-]dz) n_+ + [n_+,b_+] ) + n_-\wedge (\delbar_\Sigma n_+  + \frac{1}{2} [n_+,n_+])  \right) \\
S^\del_2(b_-+n_-+n_++b_+) & = \int_M \mr{Tr}\left( \frac{1}{2}  b_-\wedge \del_\Sigma b_- +  \frac{1}{2} b_-\wedge[b_-,b_+] + b_-\wedge[n_-,n_+] \right)
\end{align*} 
Under the isomorphism given in Theorem $\ref{BTC}$, the terms in $S^\del_2$ correspond to the three summands of the bivector field $\Pi$ above, as claimed.
\end{proof}

\begin{rmk} This Poisson structure is the $(-1)$--shifted analogue of the Poisson vertex algebra structure on the classical $\calW$ algebra $\calW_\infty(\g)$, as described in Lemma 2(b) of $\cite{Valeri}$, under the equivalence of Remark $\ref{cwa}$.
\end{rmk}

\begin{rmk}[Critical Toda Theory] For the critical level, one has to consider the Higgs version of opers $\mr{Op}_G^{ \mr{Higgs} }(\Sigma)$ inside $\mr{Higgs}_G(\Sigma)$. Accordingly, critical Toda theory is described by \[\scg_{\text{cToda}} = (\Omega^{0,\bullet}_\Sigma \otimes \n) \oplus (\Omega^{1,\bullet}_\Sigma \otimes \mathfrak{b} )[-1] \qquad \text{with} \quad l_1 = [f,-]dz \otimes \id_\g, \; l_2 = \wedge \otimes [\;,\;],  \text{ and }l_n=0\text{ for }n\geq 3 .\]
One can calculate the induced $P_0$ structure precisely as in the non-critical case.
\end{rmk}

\subsection{Kapustin--Witten Theory}

In this subsection, we study the Kapustin--Witten $\mathbb{P}^1$ of twists of $\calN=4$ supersymmetric Yang--Mills theory \cite{KapustinWitten}. In the work of the second author with C. Elliott \cite{EY}, rigorous mathematical descriptions of the A- and B-twists as classical field theories are provided. Indeed, the same method can be applied to identify all of the $\mathbb{P}^1$ of twists.

To understand the application of this field theory to the geometric Langlands program, we consider its compactification along a fixed smooth proper curve $\Sigma$. In more mathematical language, we consider the moduli space of solutions to the equations of motion for the case where the spacetime 4-manifold $X$ is of the form $X= C \times \Sigma$ with a compact Riemann surface $\Sigma$. Except for the A-twist, they are all described by $\underline{\mr{Map}}( C_{\mr{dR}} , T_\pi^*[1] \mr{Flat}_G(\Sigma)  )$, where $\pi$ is the Poisson structure of $\mr{Flat}_G(\Sigma)$ and $T_\pi^*[1] \mr{Flat}_G(\Sigma) $ is the twisted cotangent bundle, with twist given by $\pi$. The choice of twist by the Poisson structure $\pi$ of the cotangent bundle corresponds to the parameter in the family of topological twists: in particular, $\pi=0$ corresponds to the B-twist, and the A-twist is described by $\underline{\mr{Map}}( C_{\mr{dR}} , T_\Pi^*[1] \mr{Higgs}_G(\Sigma)  )$, where $\Pi$ is the Poisson structure of $\mr{Higgs}_G(\Sigma)$.

To compare with Chern--Simons theory, recall the description of Chern--Simons in terms of topological classical mechanics after compactification along $\Sigma$, that is, $\underline{\mr{Map}}( M^1_{\mr{dR}}, \mr{Flat}_G(\Sigma)  )$. Its universal bulk theory is $\underline{\mr{Map}}( M^1_{\mr{dR}} \times \R_{\geq 0} , T_\pi^*[1] \mr{Flat}_G(\Sigma)  )$, where $\pi$ is the Poisson structure of $\mr{Flat}_G(\Sigma)$. As the level $c$ of Chern--Simons theory determines the Poisson structure of $\mr{Flat}_G(\Sigma)$, we know that the universal bulk theory of Chern--Simons theory is the Kapustin--Witten twists. Moreover, the limit $\hbar \to 0$ corresponds to the limit $c \to \infty$, and hence the classical limit of Chern--Simons theory gives the B-twist. On the other hand, critical Chern--Simons theory $\underline{\mr{Map}}( M^1_{\mr{dR}}, \mr{Higgs}_G(\Sigma)  )$ has the A-twist $\underline{\mr{Map}}( M^1_{\mr{dR}} \times \R_{\geq 0} , T_\Pi^*[1] \mr{Higgs}_G(\Sigma)  )$ as the universal bulk theory.

\begin{rmk}
One might wonder if $\underline{\mr{Map}}( C_{\mr{dR}} , T ^*[1] \mr{Higgs}_G(\Sigma)  )$ also has a natural interpretation in terms of quantum field theory. Indeed, it is a version of the Kapustin twist of the 4d $\calN=4$ theory as first introduced by Kapustin in the context of $\calN=2$ theory \cite{KapustinHolo}. One has the analogy 
\[ \text{B-twist : Chern--Simons = Kapustin twist : critical Chern--Simons}.\]
\end{rmk}

Finally, let us describe the Kapustin--Witten theory in a local formal setting. For the B-twist $\underline{\mr{Map}}(X_{\mr{dR}} , T^*[3]BG )$, one has \[\scg_B = \Omega^\bullet_X \otimes (\g \ltimes \g[1]),\] where the $L_\infty$ structure on $\g \ltimes \g[1]$ is best understood from the identification $\g \ltimes \g^*[1]= \T_o[-1] (T^*[3]BG )$. For a generic twist $\underline{\mr{Map}}(X_{\mr{dR}} , T_\pi^*[3]BG )$, one has \[\scg_{\text{generic} } = \left(\xymatrix{
 \Omega^\bullet_X \otimes \g[1]  \ar[r]^-{\id}  &  \Omega^\bullet_X \otimes \g
}\right),\] where the Poisson vector becomes the identity map understood as the differential. In particular, the theory is perturbatively trivial. Similarly, for the Kapustin twist $\underline{\mr{Map}}( C_{\mr{dR}} , T ^*[1] \mr{Higgs}_G(\Sigma)  )$, one has \[\scg_{\text{Kap} } =  (\Omega^\bullet_C, d_C) \otimes (\Omega^{\bullet,\bullet}_\Sigma, \delbar_\Sigma   ) \otimes (\g \ltimes \g[1]),\] 
and for the A-twist $\underline{\mr{Map}}( C_{\mr{dR}} , T_\Pi ^*[1] \mr{Higgs}_G(\Sigma)  )$, one has \[\scg_{A } = \left(\xymatrix{
(\Omega^\bullet_C, d_C) \otimes (\Omega^{\bullet,\bullet}_\Sigma, \delbar_\Sigma   )  \otimes \g[1]  \ar[r]^-{\id}  &  (\Omega^\bullet_C, d_C) \otimes (\Omega^{\bullet,\bullet}_\Sigma, \delbar_\Sigma   ) \otimes \g
}\right),\] which again is perturbatively trivial.

\subsection{Whittaker theory}

We introduce another boundary theory of Kapustin--Witten theory. Again after compactification along $\Sigma$, Kapustin--Witten theory is described by $\underline{\mr{Map}}( C_{\mr{dR}} , T_\pi^*[1] \mr{Flat}_G(\Sigma)  )$. We consider a 3-dimensional theory $\underline{\mr{Map}}( M^1_{\mr{dR}},N_\pi^*[1] (\mr{Op}_G(\Sigma)/ \mr{Flat}_G(\Sigma)) )$, where $N_\pi^*[1] (\mr{Op}_G(\Sigma)/ \mr{Flat}_G(\Sigma))$ is the shifted conormal bundle of $\mr{Op}_G(\Sigma)$ in $T^*_\pi[1]\mr{Flat}_G(\Sigma)$. We call this theory a \emph{Whittaker theory}.

\begin{rmk}
In the work of Gaiotto--Witten \cite{GaiottoWittenKnot}, S-dual boundary condition of Chern--Simons theory is described, under the name of Nahm pole boundary conditions. Moreover, in the work of Gaitsgory \cite{GaitsgoryWhittaker}, the equivalence of factorization categories between representations of quantum group and the twisted Whittaker category was proved. We claim that it is the mathematical manifestation of the S-duality: that is, Chern--Simons theory with gauge group $G$ and level $c$ is S-dual to Whittaker theory with gauge group $\check{G}$ and level $\check{c}$, under which the duality between the categories of line defects realizes the equivalence of Gaitsgory. This is the reason why we named Whittaker theory as such. We will further investigate this theory in future work.
\end{rmk}

In the local formal setting, Whittaker theory for the B-twist, or level $c=\infty$, described by the mapping stack $\underline{\mr{Map}}( M^1_{\mr{dR}},N^*[1] (\mr{Op}_G(\Sigma)/ \mr{Flat}_G(\Sigma)) )$ becomes \[\scg_{c=\infty} = \Omega^1_{M^1} \otimes (\scg_{\text{Toda} } \ltimes \scg_{ \text{Toda} }[1] ),\] because $N^*[1] (\mr{Op}_G(\Sigma)/ \mr{Flat}_G(\Sigma)) $ can be identified with $T[1]  \mr{Op}_G(\Sigma)$ under the identification $T^*[1]\mr{Flat}_G(\Sigma) \simeq T[1] \mr{Flat}_G(\Sigma)$. For a generic level $\underline{\mr{Map}}( M^1_{\mr{dR}},N_\pi^*[1] (\mr{Op}_G(\Sigma)/ \mr{Flat}_G(\Sigma)) )$, one has \[\scg_{\text{generic} } = \left(\xymatrix{
\Omega^1_{M^1} \otimes \scg_{\text{Toda}  } [1]  \ar[r]^-{\id}  &  \Omega^1_{M^1} \otimes \scg_{\text{Toda}  }
}\right),\] which is perturbatively trivial. Also, critical Whittaker theory is described by $\underline{\mr{Map}}( M^1_{\mr{dR}},N_\Pi^*[1] (\mr{Op}^{ \mr{Higgs} }_G(\Sigma)/ \mr{Higgs}_G(\Sigma)) )$. Then the local formal description is  \[\scg_{\text{cWhit} } = \left(\xymatrix{
\Omega^1_{M^1} \otimes \scg_{\text{cToda}  } [1]  \ar[r]^-{\id}  &  \Omega^1_{M^1} \otimes \scg_{\text{cToda}  }
}\right).\]
which is again perturbatively trivial.

\newpage
\begin{appendix}

\section{Derived Deformation Theory}

Let me start with stating the fundamental theorem of derived deformation theory.

\begin{theo}
There is an equivalence of $\infty$-categories between the category Moduli of formal (pointed) moduli problems and the category dgLa of differential graded Lie algebras.
\end{theo}

The ideas of this theorem have been developed by many giants in the latter half of 20th century mathematics, including Quillen, Deligne, Drinfeld, and Feigin. Now a version of the theorem in a more general context is proved \cite{Pridham} \cite{LurieDeformation}. The articles \cite{KontsevichSoibelman} \cite{Manetti} as well as the appendix of the second volume of \cite{CostelloGwilliam} also have exposition of the result with many examples, so we only aim to provide ideas to orient the readers, referring to them for details.

Let dSt be the $\infty$-category of derived stacks. Then the situation of interest can be summarized as follows: \[\xymatrix{
\text{dSt} \ar[d]^{ (-)^\wedge_x} \ar[drr]^-{ \TT_x[-1]  } & \\
\text{Moduli}\ar@<0.5ex>[rr]^-{\Omega= \TT[-1]} &  & \text{dgLa}.\ar@<0.5ex>[ll]^-{B= \mr{MC} }
}\]

One needs to understand the following three points.
\begin{itemize}
\item If $X$ is a derived stack, then the formal neighbourhood of a point $x\in X$ is encoded by a formal moduli problem that we denote by $X^\wedge_x$.
\item There exists a functor $\Omega= \TT[-1] \colon \text{Moduli}\to \text{dgLa}$.
\item There exists a functor $B= \mr{MC} \colon \mr{dgLa}\to \text{Moduli}$.
\end{itemize}

In the context of field theory, we think of $X$ as the moduli space of solutions to the equations of motion and $X^\wedge_x$ as encoding the perturbative information around the fixed solution $x$. Then the fundamental theorem allows one to encode any perturbative information of field theory in terms of linear algebraic data.

\begin{itemize}
\item \textbf{Formal neighbourhood of a point}
\end{itemize}

Suppose we want to understand a scheme locally, namely, around a fixed point. In differential geometry, there is no way to choose an open set in a canonical way, but in algebraic geometry, we have a canonical neighbourhood around a point, which in terms of functor of points is realized by local Artinian algebras. Namely, local geometry of $X$ around $x \in X(k)$ can be completely understood by maps of the form \[\xymatrix{
\mr{Spec}(k) \ar[r]^-{x} \ar[d] & X\\
\mr{Spec}(R) \ar@{.>}[ur]
}\] for a local Artinian algebra $R \in \mr{Art}$.

In our context of derived stacks, we need to work with the category $\mr{dgArt}^{\leq 0}$ of differential graded local Artinian algebras concentrated in non-positive degrees. Then a formal moduli problem is in particular a functor $\mr{dgArt}^{\leq 0}\to \mr{sSet}$. By construction, one would have $X^\wedge_x \in \text{Moduli}$.

\begin{itemize}
\item \textbf{Shifted tangent complex}
\end{itemize}

Here is another important main character for us.

\begin{defn}
A \textit{differential graded Lie algebra} is a Lie algebra object in the category of cochain complexes, that is, a cochain complex $\g^\bullet$ with a graded anti-symmetric bracket $[\; , \; ] \colon \g^\bullet \otimes \g^\bullet \rightarrow \g^\bullet$ which is a cochain map, satisfying the Jacobi identity. More explicitly, if $\alpha \in \g^i$ and $\beta \in \g^j$, then one has $[\alpha, \beta] \in \g^{i+j}$ satisfying
\begin{itemize}
\item (anti-symmetric) $[\alpha, \beta] = -(-1)^{ij} [\beta, \alpha]$,
\item (Leibniz rule) $d[\alpha, \beta] = [d \alpha, \beta ] + (-1)^i [\alpha, d \beta]$, and
\item (Jacobi identity) $[\alpha, [\beta, \gamma]] = [[\alpha, \beta], \gamma] + (-1)^{ij} [\beta, [\alpha, \gamma]]$.
\end{itemize}
\end{defn}

Here the Jacobi identity can be thought of as ``bracketing with anything is a derivation''.

\begin{eg}
\begin{itemize}
\item An ordinary Lie algebra $\g$ is a differential graded Lie algebra concentrated at degree 0 with $d=0$. Note that a differential graded Lie algebra $\g[-i]$, concentrated at degree $i \neq 0$, should necessarily be trivial for a degree reason.
\item For a commutative differential graded algebra $(A,d_A, \cdot)$ and a differential graded Lie algebra $(\g, d_\g, [\;,\;]_\g)$, one can construct another differential graded Lie algebra $(A \otimes \g , d_{A\otimes \g} , [\;,\;]_{A \otimes \g})$ with $d_{A \otimes \g} = d_A \otimes 1 + 1 \otimes d_\g$ and $[ x \otimes \alpha,  y \otimes \beta  ]_{A \otimes\g} = (-1)^{|y| |\beta|} xy \otimes [\alpha,\beta]$. In a similar way, for instance, for a complex manifold $X$, one has a differential graded Lie algebra $( \Omega^{0,\bullet}(X, T_X^{1,0}) , \bar \del, [\;,\;] )$, where the bracket $[\;,\;]$ is induced from the one on the holomorphic tangent sheaf $T_X^{1,0}$.
\item For a smooth manifold $M$, the space $\Gamma(M,T_M)$ of tangent vectors is a Lie algebra. Let us introduce a graded vector space $T_{\mr{poly}}^\bullet(M) = \Gamma(M, \wedge^\bullet T_M)[1]$ of polyvector fields with $T^n_{\mr{poly}} (M) = \Gamma(M,\wedge^{n+1}T_M)$. We define the Schouten-Nijenhuis bracket $ [\;,\;]_{\mr{SN}} \colon T_{\mr{poly}}^k(M)\otimes T_{\mr{poly}}^l(M) \to T_{\mr{poly}}^{k+l-1}(M)$ by \[[\xi_0\wedge\cdots\wedge \xi_k ,\eta_0 \wedge \cdots \wedge \eta_l ]_{\mr{SN}} = \displaystyle\sum_{i,j} (-1)^{i+j+k}[\xi_i,\eta_j ] \wedge \xi_0 \wedge \cdots \wedge \widehat{\xi_i} \wedge \cdots \wedge \xi_k \wedge \eta_0 \wedge \cdots \wedge \widehat{\eta_j} \wedge \cdots \wedge \eta_l,\]
which makes $(T^\bullet_{\mr{poly}}(M),d=0, [\;,\;]_{\mr{SN}} )$ a differential graded Lie algebra.
\item For an associative $k$-algebra $A$, let us consider the Hochschild cochain complex $HC^\bullet(A)$, where $HC^n(A) = \mr{Hom}_k(A^{\otimes n} , A )$. For $\phi \in HC^p(A)$ and $\psi \in HC^q(A)$, we define a non-associative product, called the Gerstenhaber product, $\phi \circ \psi \in HC^{p+q-1}(A)$ by \[(\phi\circ \psi)(a_1,\cdots, a_{p+q-1} ) = \displaystyle\sum_{i=1}^{p} (-1)^{(i-1)(q-1)} \phi(a_1,\cdots, a_{i-1}, \psi(a_i, \cdots, a_{i+q-1} ) , a_{i+q} , \cdots, a_{p+q-1} ),\] and the Gerstenhaber bracket $[\phi,\psi]_{G} \in HC^{p+q-1}(A)$ by $[\phi,\psi]_G = \phi \circ \psi - ( -1)^{(p-1)(q-1)} \psi \circ \phi$. Also, for the multiplication operator $\mu \colon HC^2(A)$, we define the differential $d_\mu=[\mu,- ] \colon HC^n(A) \rightarrow HC^{n+1}(A) $. One can check that $(HC^\bullet(A)[1] , d_\mu, [\;,\;]_{G} )$ is a differential graded Lie algebra.
\end{itemize}
\end{eg}

\begin{rmk}
\begin{itemize}
\item $( \Omega^{0,\bullet}(X, T_X^{1,0}) , \bar \del, [\;,\;] )$ controls deformations of $X$ as a complex manifold.
\item $(T^\bullet_{\mr{poly}}(M),d=0, [\;,\;]_{\mr{SN}} )$ controls deformations of Poisson structures on $M$, while $(HC^\bullet(A)[1] , d_\mu, [\;,\;]_{G} )$ controls deformations of $A$ as an associative algebra. If $A=\calO_M$ for a smooth manifold $M$, then one defines $D^\bullet_{\mr{poly}}(M)$ to be a subcomplex of $HC^\bullet(\calO_M)[1]$ consisting of polydifferential operators. Kontsevich constructed an explicit equivalence between $T^{\bullet}_{\mr{poly}}(M)$ and $D^{\bullet}_{\mr{poly}}(M)$ in dgLa, which is the single most important ingredient for his celebrated deformation quantization result for a real Poisson manifold.
\end{itemize}
\end{rmk}

Now the claim is that for a derived stack $X$ and a point $x\in X$, the shifted tangent complex $\mathbb{T}_x[-1]X$ has the structure of a differential graded Lie algebra. First let us give an abstract reasoning. Consider $\Omega_x X := x \times_X x$ which is a group object in the category of derived stacks. In derived algebraic geometry, one can always take the Lie algebra of a group object to get $\mr{Lie}(\Omega_x X ) =: \TT_x X[-1]$ which is necessarily a differential graded Lie algebra.

\begin{rmk}
For the skyscraper sheaf $k_x$ at $x\in X$, one has $\R\mr{Hom}_{\calO_X}(k_x, k_x ) = U( \TT_x X[-1] )$ where $U( \TT_x X[-1] )$ is the universal enveloping algebra of the Lie algebra object $\TT_x X[-1]$ in the category of chain complexes. This gives a sense in which taking the shifted tangent complex is an instance of the Koszul duality.
\end{rmk}

A formal moduli problem has a unique point, so this explains what we mean by the shifted tangent complex $\TT_F[-1]$. On the other hand, this definition is hardly useful for actually performing any computation. Let us list some important examples: we phrase the result of shifted tangent complexes for global moduli spaces with a specified point rather than formal moduli problems.

\begin{eg}
\begin{itemize}
\item The classifying stack $BG$ for a group $G$ has $\TT_o [-1] BG = \g$ with its Lie algebra structure.
\item The mapping stack $\underline{\mr{Map}}( X , Y )$ has $\TT_{f} \underline{\mr{Map}}( X , Y ) = \R\Gamma(X, f^*\TT_Y )$. In particular, for a (compact oriented) manifold $M$ and for a (smooth proper) complex algebraic variety $X$,
\begin{itemize}
\item $\mr{Loc}_G(M) : = \underline{\mr{Map}}(M_B, BG)$ has $\TT_{\text{triv} }[-1] \mr{Loc}_G(M)  = (C^\bullet( M ,\g ), d_M, [\;,\;]) $, where the bracket $[\;,\;]$ is induced from the one on $\g$.
\item $\mr{Bun}_G(X) : = \underline{\mr{Map}}(X, BG)$ has $\TT_{\text{triv} }[-1] \mr{Bun}_G(X)   = (\Omega^{0,\bullet}(X,\g) , \bar \del_X , [\;,\;]) $.
\item $\mr{Flat}_G(X) :=  \underline{\mr{Map}}( X_{\mr{dR}} , BG )$ has $\TT_{\text{triv} }[-1] \mr{Flat}_G(X)   =( \Omega^{\bullet}(X,\g), d_X, [\;,\;]  )$.
\item $\mr{Higgs}_G(X) :=  \underline{\mr{Map}}( X_{\mr{Dol}} , BG )$ has $\TT_{\text{triv} }[-1] \mr{Bun}_G(X)  = ( \Omega^{\bullet,\bullet}(X,\g) ,  \bar \del_X, [\;,\;] ) $.
\end{itemize}
Each shifted tangent complex controls deformations of such $G$-bundles.
\end{itemize}
\end{eg}

\begin{rmk}
As we work with the $\infty$-category of differential graded Lie algebras, there is no actual difference between differential graded Lie algebras and $L_\infty$ algebras. On the other hand, for our purpose toward describing field theory, it is more convenient to work with a strict model of $L_\infty$ algebras and hence we are going to do so.
\end{rmk}

\begin{itemize}
\item\textbf{Maurer--Cartan functor}
\end{itemize}

It remains to understand the quasi-inverse functor we denote by MC.

\begin{defn}
Let $\g$ be a differential graded Lie algebra. The \textit{Maurer--Cartan functor} $\mr{MC}_\g \colon \mr{dgArt}^{\leq 0} \rightarrow \mr{sSet}$ is defined by $(R,\mm) \mapsto \mr{MC}_\g(R)$ with \[\mr{MC}_\g(R)[n] := \{ \alpha \in \left( \mm \otimes \Omega^\bullet(\Delta^n) \otimes \g \right)^1 \mid d \alpha + \frac{1}{2} [\alpha, \alpha] = 0 \}.\]
\end{defn}

It is a nontrivial theorem \cite{Hinich} \cite{Getzler} that $\mr{MC}_\g$ defines a formal moduli problem, which we sometimes write as $B\g$. Moreover, the fundamental theorem says that any formal moduli problem arises in this way, up to homotopy.

In fact, one can represent the formal moduli problem $B\g$ in a more explicit way, when $\mr{dim}(\g^i)<\infty$ for each $i$. That is, not only do we know $\mr{MC}_\g(-) =  \mr{Map}_{\mr{dSt}_\ast}(-,X^\wedge_x )$ by the fundamental theorem, we have $\mr{MC}_\g(-)= \mr{Map}_{\mr{cdga}_\ast}(A,-)$ for some $A\in \mr{cdga}$, where $A$ does not necessarily lie in $\mr{cdga}^{\leq 0}$ because the corresponding object $X$ possibly has some stacky nature.

\begin{defn}
The \emph{Chevalley--Eilenberg cochain complex} $C^\bullet(\g)$ of a differential graded Lie algebra $\g$ is an augmented commutative differential graded algebra defined as follows: as a graded algebra one has \[C^\bullet(\g) = \widehat{\mr{Sym}}( \g^*[-1] ) = \displaystyle\prod_{i \geq 0} \mr{Sym} ^i(\g^*[-1])\] and the differential $d_{\mr{CE}} =d_\g + d_{[\;,\;]} $ is defined as a derivation, where on generators $\g^*[-1]$, the differential $d_\g \colon \g^*[-1] \rightarrow \g^*[-1] $ is defined as the dual of the differential $d \colon \g \rightarrow \g$ and the differential $d_{[\;,\; ] }\colon \g^*[-1] \rightarrow \mr{Sym}^2 (\g^*[-1]) = \wedge^2(\g^*)[-2]$ is defined as the dual of the bracket map $[\;,\;] \colon \wedge^2 \g \rightarrow \g$.
\end{defn}

We should think of $(C^\bullet(\g), d_{\mr{CE}}) $ as the structure sheaf of the formal moduli problem $B\g$.

\begin{rmk}
The differential $d_{\mr{CE}}$ can be regarded as a vector field of cohomological degree 1 on $\widetilde{B\g}$ with $\calO(\widetilde{B\g})= ( \widehat{\mr{Sym}}( \g^*[-1] )  , d =0)$, that is, from $\TT_{\widetilde{B\g}} = \g[1]$, one has \[d_{\mr{CE}} \in \Gamma(\widetilde{B\g}, \TT_{\widetilde{B\g}}[1])=   \displaystyle\prod_{k \geq 0 } \mr{Hom}( \g[1]^{\otimes k} , \g[2] )_{S_k} \subset \displaystyle\prod_{k\geq 0} \mr{Hom}( \g^{\otimes k} , \g[2-k]  ). \] In fact, one can take a definition of (possibly curved) $L_\infty$ algebra structure on $\g$ as such a vector field $Q_\g$ of cohomological degree 1: the equation $Q_\g^2=0$ corresponds to the $L_\infty$ equations. 
\end{rmk}
\end{appendix}

\emph{Email address}: \href{mailto:dbutson@perimeterinstitute.ca}{dbutson@perimeterinstitute.ca}, \href{mailto:philsang@math.northwestern.edu}{philsang@math.northwestern.edu}


\begin{thebibliography}{AG12}

\bibitem[AKSZ97]{AKSZ}
Mikhail Alexandrov, Maxim Kontsevich, Albert Schwarz, and Oleg Zaboronsky.
\newblock {\em The geometry of the master equation and topological field theory}
\newblock Internat. J. Modern Phys. 12 (1997), no. 7, 1405--1429.

\bibitem[AB67]{AtiyahBott}
Michael Atiyah and Raoul Bott.
\newblock {\em A Lefschetz Fixed Point Formula for Elliptic Complexes: I}.
\newblock Ann. of Math. (2), 86 (1967), no. 2, 374--407.

\bibitem[BD04]{BeilinsonDrinfeld}
Alexander Beilinson and Vladimir Drinfeld.
\newblock {\em Chiral algebras}.
\newblock American Mathematical Society Colloquium Publications, vol. 51, American Mathematical Society, Providence, RI, 2004.

\bibitem[BD05]{BDopers}
Alexander Beilinson and Vladimir Drinfeld.
\newblock {\em Opers}.
\newblock preprint arXiv:math/0501398.

\bibitem[BCOV94]{BCOV}
Michael Bershadsky, Sergio Cecotti, Hirosi Ooguri, and Cumrun Vafa.
\newblock {\em Kodaira--Spencer theory of gravity and exact results for quantum string amplitudes}.
\newblock Comm. Math. Phys. 165 (2), 311--427 (1994).

\bibitem[CF00]{CattaneoFelder} 
Alberto Cattaneo and Giovanni Felder.
\newblock {\em A path integral approach to the Kontsevich quantization formula}.
Commun. Math. Phys. 212, 591--611 (2000).

\bibitem[CF01]{CattaneoFelderII} 
Alberto Cattaneo and Giovanni Felder.
\newblock {\em Poisson sigma models and deformation quantization}.
Mod. Phys. Lett. A 16, 179--190 (2001).

\bibitem[CMR14]{CattaneoMnevReshetikhin}
Alberto Cattaneo, Pavel Mnev, and Nicolai Reshetikhin.
\newblock {\em Classical BV theories on manifolds with boundary}.
Commun. Math. Phys. 332 2, 535--603 (2014).

\bibitem[CMR16]{CattaneoMnevReshetikhinII}
Alberto Cattaneo, Pavel Mnev, and Nicolai Reshetikhin.
\newblock {\em Perturbative quantum gauge theories on manifolds with boundary}.
\newblock preprint arXiv:1507.01221. {\em To appear in Commun. Math. Phys.}

\bibitem[CLL15]{ChanLeungLi}
Kwokwai Chan, Naichung Conan Leung, and Qin Li.
\newblock {\em A mathematical foundation of Rozansky--Witten theory}.
\newblock preprint arXiv:1502.03510.

\bibitem[Cos11a]{CostelloBook}
Kevin Costello.
\newblock {\em Renormalization and effective field theory}.
\newblock Mathematical Surveys and Monographs, AMS, 2011, 170.

\bibitem[Cos11b]{CostelloWittenGenusII}
Kevin Costello.
\newblock {\em A geometric construction of the Witten genus, II}.
\newblock preprint arXiv:1112.0816.

\bibitem[Cos13]{CostelloYangian}
Kevin Costello.
\newblock {\em Supersymmetric gauge theory and the Yangian}.
\newblock preprint arXiv:1303.2632.

\bibitem[Cos16]{CostelloM-theory}
Kevin Costello.
\newblock {\em M-theory in the $\Omega$-background and 5-dimensional non-commutative gauge theory}.
\newblock preprint arXiv:1610.04144.

\bibitem[CG16]{CostelloGwilliam}
Kevin Costello and Owen Gwilliam.
\newblock {\em Factorization algebras in quantum field theory}.
\newblock 2016.
\newblock two-volume book project in progress, available at
  \url{http://people.mpim-bonn.mpg.de/gwilliam/}.

\bibitem[CL12]{CostelloLi1}
Kevin Costello and Si Li.
\newblock {\em Quantum BCOV theory on Calabi--Yau manifolds and the higher genus B-model}.
\newblock preprint arXiv:1201.4501.

\bibitem[CL15]{CostelloLi2}
Kevin Costello and Si Li.
\newblock {\em Quantization of open-closed BCOV theory, I}.
\newblock preprint arXiv:1505.06703.

\bibitem[DSKV13]{DSKV}
Alberto De Sole, Victor G. Kac, and Daniele Valeri.
\newblock {\em Classical W-algebras and generalized Drinfeld--Sokolov bi-Hamiltonian systems within the theory of Poisson vertex algebras}
\newblock Comm. Math. Phys. 323 (2013), n. 2, 663--711.

\bibitem[DS81]{DrinfeldSokolovI}
Vladimir Drinfeld and Vladimir Sokolov
\newblock {\em Equations of Korteweg--de Vries type and simple Lie algebras}.
\newblock Soviet Mathematics Doklady, vol. 23 (1981), No. 3, p. 457--462.

\bibitem[DS85]{DrinfeldSokolovII}
Vladimir Drinfeld and Vladimir Sokolov
\newblock {\em Lie algebras and equations of Korteweg--de Vries type}.
\newblock Journal of Soviet Mathematics, vol. 30 (1985), p. 1975--2035.

\bibitem[EY15]{EY}
Chris Elliott and Philsang Yoo.
\newblock {\em Geometric Langlands twists of $N=4$ gauge theory from derived algebraic geometry}.
\newblock preprint arXiv:1507.03048.

\bibitem[GW09]{GaiottoWittenBoundary}
Davide Gaiotto and Edward Witten.
\newblock {\em Supersymmetric boundary conditions in $N=4$ super Yang-Mills theory}.
J. Statist. Phys. 135 (2009) 789--855.

\bibitem[GW12]{GaiottoWittenKnot}
Davide Gaiotto and Edward Witten.
\newblock {\em Knot invariants from four-dimensional gauge theory}.
\newblock Adv. Theor. Math. Phys. 16 (2012), 935--1086.

\bibitem[Gait08]{GaitsgoryWhittaker}
Dennis Gaitsgory.
\newblock {\em Twisted Whittaker model and factorizable sheaves}.
\newblock Selecta Math., 13 (2008), 617--659.

\bibitem[Get09]{Getzler}
Ezra Getzler.
\newblock {\em Lie theory for nilpotent $L_\infty$-algebras}.
\newblock Ann. of Math. (2) 170 (2009), no. 1, 271--301.

\bibitem[GGW16]{GGW}
Vassily Gorbounov, Owen Gwilliam, and Brian Williams.
\newblock {\em Chiral differential operators via Batalin--Vilkovisky quantization}.
\newblock preprint arXiv:1610.09657.

\bibitem[GG14]{GradyGwilliam}
Ryan Grady and Owen Gwilliam.
\newblock {\em One-dimensional Chern--Simons theory and the $\hat{A}$ genus}.
\newblock Algebr. Geom. Topol., 14 (2014), 2299--2377.

\bibitem[GG15]{GradyGwilliamLinfinity}
Ryan Grady and Owen Gwilliam.
\newblock {\em $L_\infty$ spaces and derived loop spaces}.
\newblock New York J. Math. 21 (2015), 231--272.

\bibitem[GLL15]{GradyLiLi}
Ryan Grady, Qin Li, and Si Li.
\newblock {\em BV quantization and the algebraic index}.
\newblock preprint arXiv:1507.01812.

\bibitem[Hin01]{Hinich}
Vladimir Hinich.
\newblock {\em DG coalgebras as formal stacks}.
\newblock J. Pure Appl. Algebra 162 (2001), no. 2-3, 209--250.

\bibitem[JF16]{Theo}
Theo Johnson-Freyd
\newblock {\em Exact triangles, Koszul duality, and coisotopic boundary conditions}.
\newblock preprint arXiv:1608.08598

\bibitem[Kap06]{KapustinHolo}
Anton Kapustin.
\newblock {\em Holomorphic reduction of $N=2$ gauge theories, Wilson--'t Hooft operators, and S-duality}.
\newblock preprint arXiv hep-th/0612119.

\bibitem[KW07]{KapustinWitten}
Anton Kapustin and Edward Witten.
\newblock {\em Electric-magnetic duality and the geometric Langlands program}.
\newblock Communications in Number Theory and Physics. 1 (2007), no. 1, 1--236.

\bibitem[Kon03]{Kontsevich}
Maxim Kontsevich.
\newblock {\em Deformation quantization of Poisson manifolds, I}.
\newblock Lett. Math. Phys. 66 (2003) 157.

\bibitem[KS]{KontsevichSoibelman}
Maxim Kontsevich and Yan Soibelman.
\newblock {\em Deformation theory, volume I}.
\newblock available at 
\url{https://www.math.ksu.edu/~soibel/Book-vol1.ps}

\bibitem[LL16]{LiLi}
Qin Li and Si Li.
\newblock {\em On the B-twisted topological sigma model and Calabi--Yau geometry}.
\newblock J. Diff. Geom. 102 (2016) no.3, 409--484.

\bibitem[Li16]{LiVAQME}
Si Li.
\newblock {\em Vertex algebras and quantum master equation}.
\newblock preprint arXiv:1612.01292.

\bibitem[Lur08]{LurieCobordism}
Jacob Lurie.
\newblock {\em On the classification of topological field theories}.
\newblock Current developments in mathematics, 2008, Int. Press, Somerville, MA, 2009, pp. 129--280.

\bibitem[Lur11]{LurieDeformation}
Jacob Lurie.
\newblock {\em Derived Algebraic Geometry X: Formal Moduli Problems}.
\newblock Available at
\url{http://www.math.harvard.edu/~lurie/papers/DAG-X.pdf}.

\bibitem[Man09]{Manetti}
Marco Manetti.
\newblock {\em Differential graded Lie algebras and formal deformation theory}.
\newblock Algebraic geometry-- Seattle 2005. Part 2, Proc. Sympos. Pure Math., vol. 80, Amer. Math. Soc., Providence, RI, 2009, pp. 785--810.

\bibitem[MS16]{MelaniSafronov}
Valerio Melani and Pavel Safronov.
\newblock {\em Derived coisotropic structures}.
\newblock preprint arxiv:1608.01482.

\bibitem[PTVV13]{PTVV}
Tony Pantev, Bertrand To{\"e}n, Michel Vaqui{\'e}, and Gabriele Vezzosi.
\newblock {\em Shifted symplectic structures}.
\newblock Publications math{\'e}matiques de l'IH{\'E}S, 117 (2013), no. 1, 271--328.

\bibitem[Pri10]{Pridham}
Jonathan Pridham.
\newblock {\em Unifying derived deformation theories}.
\newblock Adv. Math. 224 (2010), no.3, 772--826.

\bibitem[Roz17]{Rozenblyum}
Nick Rozenblyum.
\newblock {\em Higher Poisson centers and traces}.
\newblock {\em In preparation}

\bibitem[Saf15]{Safronov}
Pavel Safronov.
\newblock {\em Poisson reduction as a coisotropic intersection}.
\newblock preprint arXiv:1509.08081.

\bibitem[Sch93]{Schwarz}
Albert Schwarz.
\newblock {\em Geometry of Batalin--Vilkovisky quantization}.
\newblock Comm. Math. Phys. 155 (1993), no. 2, 249--260.

\bibitem[Val13]{Valeri}
Daniele Valeri.
\newblock {\em Classical W-algebras within the theory of Poisson vertex algebras}.
\newblock Advances in Lie Superalgebras. Edited by M. Gorelik and P. Papi. Springer INdAM Series Vol 7, 2013. pp. 203--221.

\end{thebibliography}
\end{document}